\documentclass[11pt]{article}
\usepackage[utf8]{inputenc}
\usepackage[margin=1in]{geometry}
\usepackage{lmodern}

\usepackage{microtype}
\usepackage{graphicx}
\usepackage{subfigure}
\usepackage{booktabs}
\usepackage[round]{natbib}

\usepackage{amsmath,amsthm,amsfonts,amssymb,mathdots,array,mathrsfs,bm,bbm,stmaryrd,xcolor}

\usepackage{breakcites}

\usepackage[T1]{fontenc}
\usepackage{enumerate}
\usepackage{inputenc}

\usepackage{graphicx}
\usepackage{wrapfig}

\usepackage{algpseudocode}
\usepackage{algorithm}

\usepackage{color}

\usepackage[utf8]{inputenc} 
\usepackage[T1]{fontenc}    

\usepackage[colorlinks,linkcolor=red,filecolor=blue,citecolor=blue,urlcolor=blue]{hyperref}

\usepackage{url}            
\usepackage{amsfonts}       
\usepackage{nicefrac}       
\usepackage{microtype}      
\usepackage{xcolor}
\usepackage{enumitem}
\usepackage{bbm}

\usepackage{tikz}
\usetikzlibrary{mindmap}

\usepackage{tikzpagenodes}

\usepackage[toc,page]{appendix}

\usepackage{math_notation}

\newtheorem*{theorem*}{Theorem}

\newtheorem{theorem}{Theorem}
\newtheorem{lemma}{Lemma}
\newtheorem{corollary}{Corollary}

\newtheorem{remark}{Remark}

\newenvironment{proofoutline}
 {\proof[Proof outline]}
 {\endproof}

\def \pset {\calP}

\def \hardthres {\mathsf{thres}}
\def \thres {\mathsf{thres}}

\def \dh {\mathsf{d_H}}
\def \dH {\mathsf{d_H}}

\def \bbetatrue {\bbeta^{\natural}}
\def \betatrue {\beta^{\natural}}

\def \belambda {\lambda_{\bbeta}}
\def \xilambda {\lambda_{\bXi}}

\def \bPitrue {\bPi^{\natural}}

\def \entH {\mathsf{H}}
\def \enth {\mathsf{h}}
\def \entI {\textup{I}}
\def \dh {\mathsf{d_H}}

\definecolor{celadon}{rgb}{0.67, 0.88, 0.69}
\definecolor{chromeyellow}{rgb}{1.0, 0.65, 0.0}
\definecolor{columbiablue}{rgb}{0.61, 0.87, 1.0}

\allowdisplaybreaks

\begin{document}

\title{\bf Sparse Recovery with Shuffled Labels:\\ Statistical Limits and Practical Estimators\vspace{0.3in}}

\author{
  \textbf{Hang Zhang, \ \  Ping Li} \vspace{0.1in}\\
  Cognitive Computing Lab\\
  Baidu Research\\
  10900 NE 8th ST. Bellevue, WA 98004, USA\\
  \texttt{ \{zhanghanghitomi,\ pinli98\}@gmail.com}
}
\date{\vspace{0.1in}}
\maketitle

\begin{abstract}\vspace{0.2in}

\noindent\footnote{Preliminary work appeared in Proceedings of the IEEE International Symposium on Information Theory (ISIT'21).}This paper considers the
sparse recovery with shuffled labels, i.e.,
$\by = \bPitrue \bX \bbetatrue + \bw$,
where $\by \in \RR^n$, $\bPi\in \RR^{n\times n}$, $\bX\in \RR^{n\times p}$,
$\bbetatrue\in \RR^p$, $\bw \in \RR^n$ denote
the sensing result, the unknown permutation matrix, the
design matrix, the sparse signal, and
the additive noise, respectively.
Our goal is to reconstruct both the permutation matrix
$\bPitrue$
and the sparse signal $\bbetatrue$.
We investigate this problem from both
the statistical and computational aspects.
From the statistical aspect,
we first establish the minimax lower bounds on the sample number $n$
and the \emph{signal-to-noise ratio} ($\snr$) for
the correct recovery of
permutation matrix $\bPitrue$ and the support set
$\supp(\bbetatrue)$,
to be more specific,
$n \gtrsim k\log p$ and $\log\snr \gtrsim \log n + \frac{k\log p}{n}$.
Then, we confirm the tightness of these minimax lower bounds by presenting
an exhaustive-search based estimator whose performance
matches the lower bounds thereof up to some multiplicative constants.
From the computational aspect,
we impose a parsimonious assumption on the number of permuted rows
and propose a computationally-efficient estimator accordingly. Moreover, we show that
our proposed estimator can obtain the ground-truth $(\bPitrue, \supp(\bbetatrue))$ under mild conditions.
Furthermore, we provide numerical experiments
to corroborate
our claims.

\vspace{0.2in}
\noindent In this study, we focus on the ``single measurement'' problem, i.e., $\by \in \RR^n$ and
$\bbetatrue\in \RR^p$, and require $\snr$ being at least of the
order $\Omega(n^c \cdot (k\cdot \log p/n)^{c})$. A recent work~\citep{zhang2023one} studies the permuted sparse recovery problem with multiple measurements,  i.e., $\by \in \RR^{n\times m}$ and
$\bbetatrue\in \RR^{p\times m}$, with $m>1$. They exploit the strategy of
``borrowing strength`` across different sets of measurements
and reduce the $\snr$ requirement for the permutation recovery. They propose a new estimator and develop a series of new techniques (including a novel modification of the ``leave-one-out'' method), which however do not apply to our problem  (with $m=1$) in this paper.

\end{abstract}

\newpage

\section{Introduction}
This paper considers the linear sensing relation with
shuffled labels, which can be formulated as
\[
\by = \bPitrue\bX\bbetatrue + \bw.
\]
Here $\by\in \RR^n$ denotes the sensing result,
$\bPitrue \in \RR^{n\times n}$ is the unknown permutation matrix,
$\bX \in \RR^{n\times p}$ is the design (sensing) matrix, $\bbetatrue \in \RR^p$
represents the signals of interests,
and $\bw \in \RR^n$ denotes the additive noise.
In real life, we have witnessed a broad spectrum of its applications,
which spans from communication to data privacy
to
computer vision to curve registration to natural language processing~\citep{pananjady2016linear, slawski2017linear, unnikrishnan2015unlabeled, zhang2019permutation}.
Two prominent examples include \emph{linkage record}, which
is to merge two datasets
containing different pieces of information
about the same objects
into one comprehensive
dataset; and \emph{data de-anonymization}, which infers the hidden labels and
can be viewed as the inverse problem
of data anonymization used for privacy protection.
Apart from the above applications,
other applications include
correspondence estimation between
pose and estimation in
graphics to time-domain sampling in
the presence of clock jitter to multi-target tracking in radar.
For a detailed discussion, we
refer the interested readers to~\citet{pananjady2016linear, slawski2017linear, unnikrishnan2015unlabeled, zhang2019permutation, slawski2020two}.

In the majority of  existing works~\citep{pananjady2016linear, slawski2017linear, unnikrishnan2015unlabeled, zhang2019permutation, slawski2020two,   hsu2017linear, zhang2022benefits, zhang2020optimal},
their focus is usually  on the regime of sufficient samples, in other words, the sample number $n$ is larger than the signal length $p$ (i.e., $n\geq p$). For a general case where the signal $\betatrue \in \RR^p$
an arbitrary vector residing within the linear space $\RR^p$, the requirement $n\geq $ seems to be inevitable, even with the perfect correspondence information, namely, the permutation matrix $\bPitrue$. Meanwhile, the sample number can be reduced given some prior
knowledge of the signal $\bbetatrue$, e.g., we know that $\bbetatrue$ lies within a small subspace, or equivalently,
$\bbetatrue$ is with a low inherent dimension.
One typical example is the literature on the ``compressed sensing'' (or \emph{sparse recovery})~\citep{donoho2006compressed, candes2006stable, candes2006robust}.
Assuming the signal
$\bbetatrue$ is $k$-sparse (with $k$-nonzero entries), it is proved that the required sample number $n$ can be reduced to
$\Omega(k\log p)$, which is far less than $p$ provided $k \ll p$.
For other low-dimensional
structures,
similar results have been obtained under the names
\emph{M-estimator with regularizer}s~\citep{negahban2012unified},
\emph{atomic norm}s~\citep{chandrasekaran2012convex}, \emph{random convex optimization}s~\citep{amelunxen2014living}, etc.

Inspired by these works, we
investigate the shuffled linear sensing problem
with insufficient samples, namely, $n\ll p$, by placing
a parsimonious assumption on the signal $\bbetatrue$.
Assuming $\bbetatrue$ to be $k$-sparse, we
show that the correspondence information can be restored
with $n = \Omega(k\log p)$.
Notice that this order
is the same as the classical works on compressed sensing/sparse recovery~\citep{donoho2006compressed, candes2006stable, candes2006robust} and is far less than the previously required sample number such that $n=\Omega(p)$~\citep{pananjady2016linear, slawski2017linear, hsu2017linear, unnikrishnan2015unlabeled}.

\vspace{0.1in}

\noindent\textbf{Related work.}
The research on unlabeled linear regression has a long history
and can be at least traced back to 70s under the name
``broken sample problems''
\citep{degroot1976matching, goel1975re, bai2005broken, degroot1980estimation}. In recent years, we have witnessed
a renaissance of the study in this area.
\citep{unnikrishnan2015unlabeled} investigate
 the permutation recovery under the noiseless setting,
 i.e., $\bw = \bZero$;
 and establish the necessary condition
 $n\geq 2p$ for the general signal recovery.
 In~\citet{pananjady2016linear},
the noisy observation is considered and a thorough analysis of the \emph{maximum likelihood} (ML) estimator is presented.
From the statistical perspective,
it is shown that the ML estimator can reach the
statistical optimality with respect to the \emph{signal-to-noise ratio} ($\snr \defequal \|\bbetatrue\|^2_2/\sigma^2$)
requirement for correct permutation recovery,
to be more specific, $\snr \asymp \Omega(n^c)$.
From the computational perspective,
\citet{pananjady2016linear} show the ML estimator is NP-hard except for the special case
$p =1$.
This computational issue is later tackled by~\citet{hsu2017linear},
where an approximation algorithm for the permutation recovery
is proposed with polynomial complexity.
In~\citet{slawski2017linear}, the authors take an alternative path
and impose parsimonious constraints on the number of permuted rows.
By viewing $(\bI-\bPitrue)\bX\bbetatrue$ as sparse outliers,~\citep{slawski2017linear} reconstruct the correspondence information from the viewpoint of de-noising. In this work, we adopt a similar
viewpoint in designing the estimator for practical usage. However, some
modifications are required to handle the insufficient sample problem (i.e., $n \ll p$). A detailed explanation of our proposed estimator
can be found in Section~\ref{sec:prac_estim}.

In~\citet{emiya2014compressed}, they
consider a similar setting as ours, namely, a sparse signal $\bbetatrue$.
A branch-and-bound scheme is proposed to reconstruct the permutation matrix $\bPitrue$. Potential drawbacks of this work include their high computational cost and the missing performance guarantee.
\citep{zhang2021sparse} is the
conference version of this work
and proposes different practical estimators.
While the estimators in~\citet{zhang2021sparse} are rooted in the literature
about the sign consistency in Lasso, Dantzig estimator, etc~\citep{zhao2006model, wainwright2009sharp, meinshausen2009lasso, donoho2005stable, rosenbaum2010sparse,  zhang2017quadratic, lounici2008sup, zhang2018sparse},
the estimator in this work is more related
to the study of robustness~\citep{nguyen2013robust, dalalyan2019outlier}.
Despite the above differences and their
distinct looks, we should mention that
they actually share the same spirit, i.e., the viewpoint of de-noising:
\citet{zhang2021sparse} performs de-noising
in an implicit way while this work takes an explicit approach. Together with the change brings a noticeable performance improvement, which
is discussed in Remark~\ref{remark:performance_improve}.

Apart from the above-mentioned articles,
there are other works that are worth mentioning,
e.g.,~\citep{tsakiris2019homomorphic, haghighatshoar2018signal, emiya2014compressed, zhang2019permutation,  slawski2020two, fang2022regression, slawski2022permuted}. Since their connection to
our work are rather loose, we only mention their name without giving
detailed discussion.

\vspace{0.1in}

\par
\noindent \textbf{Contributions.}\
Our contributions are summarized as follows:
\begin{itemize}
\item
We establish the statistical lower bounds
for the correct recovery of $(\bPitrue, \supp(\bbetatrue))$.
Different from the previous works, our work
focuses on the situation with insufficient samples,
 i.e., $n \ll p$. Exploiting the signal sparsity,
 we manage to reduce the sample number $n$ from
 $\Omega(p)$ to $\Omega(k\log p)$, where
 $k$ denotes the sparsity number of the signal $\bbetatrue$.
As compensation, our required $\snr$ inflates
from $\log\snr \gsim \log n$ to
$\log \snr \gsim \log n + \frac{k\log p}{n}$, which
turns out to be marginal since $n\gsim k\log p$.
Moreover, we show an
exhaustive-search-based estimator can
match the above lower bounds up to some multiplicative constants and thus conclude the tightness of the minimax lower bounds thereof.
\item
We propose a computational-friendly estimator for the recovery
of $(\bPitrue, \supp(\bbetatrue))$.
By imposing a parsimonious assumption
on the number of permuted rows, we view $(\bI - \bPitrue)\bX\bbetatrue$
as a sparse outlier and obtain a rough estimation $\wt{\bbeta}$ of the
signal $\bbetatrue$. Then, we reconstruct the missing correspondence based on the estimated value $\wt{\bbeta}$.
We prove that the ground-truth
permutation matrix $\bPitrue$ can be obtained under mild conditions.
More importantly, we show these conditions almost match the
minimax lower bounds thereof.
Once the permutation matrix $\bPitrue$ is given, we restore our
problem to the classical setting of sparse recovery/compressed sensing and detect the support set of $\bbetatrue$ accordingly.
\end{itemize}

\newpage

\noindent \textbf{Notations.}\
We denote $c, c_0, c^{'}$ as some positive real constants.
For two arbitrary real numbers, we
denote $a\vcup b$ as
the maximum of $a$ and $b$
while $a\vcap b$ as the minimum.
We denote $a\lsim b$ if there exists
a positive constant $c_0 > 0$ such that
$a\leq c_0 b$. Similarly, we define
$a\gsim b$ provided the inequality $a\geq c_0 b$ holds for
some positive constants $c_0$.
We write $a\asymp b$ when
$a\lsim b$ and $a\gsim b$ hold simultaneously.

We denote the set of possible values for the
permutation matrix $\bPitrue$ as
$\pset_n$.
For an arbitrary permutation matrix $\bPi$, we associate it
with the operator $\pi(\cdot)$ which transforms index $i$
to $\pi(i)$.  We  define the Hamming distance
$\dH(\cdot, \cdot)$ between two permutation matrices
$\bPi_1$ and $\bPi_2$ as
$\dH\bracket{\bPi_1, \bPi_2} \defequal \sum_{i=1}^n \Ind\bracket{\pi_1(i) \neq \pi_2(i)}$.
The support
set $\supp(\bbetatrue)$ is defined
as the set of indices of non-zero entries (i.e.,
$\supp(\bbetatrue)\defequal \{i: \beta^{\natural}_i\neq 0\}$).
The
\emph{signal-to-noise-ratio} ($\snr$) is defined as
${\|\bbetatrue\|_2^2}/{\sigma^2}$.


\section{Problem Statement}\label{sec:sys_mdl}
We start by giving a formal restatement of our problem.
Consider the sensing relation
\begin{align}
\label{eq:sys_mdl}
  \by = \bPitrue\bX\bbetatrue + \bw,
\end{align}
where $\by \in \RR^n$ is the sensing result,
$\bPitrue\in \RR^{n\times n}$ denotes the
permutation matrix, i.e., $\sum_{i} \bPitrue_{ij} =
\sum_{j}\bPitrue_{ij} = 1$, $\bPitrue_{ij}\in \set{0, 1}$,
$\bX \in \RR^{n\times p}$ is the sensing matrix, with
its entries to be i.i.d. standard normal random variable, namely,
$\bX_{ij}\in \normdist(0, 1)$,
$\bbetatrue\in \RR^{p}$ represents the $k$-sparse signals, i.e.,
$\|\bbetatrue\|_0 \leq k$,
and $\bw \in \RR^n$ denotes the Gaussian noise following
$\normdist(\bZero, \sigma^2 \bI)$.

Compared with the previous work as in~\citet{zhang2019permutation, slawski2020two, pananjady2016linear} that requires
$n\geq 2p$, our work
focuses on the regime where $n \leq p$. By exploiting the
sparsity of the signals $\bbetatrue$, we will show
that $n = \Omega\bracket{k\log p} \ll p$ samples will be sufficient
for the recovery of the permutation matrix.
A graphical illustration of this paper's organization is
presented in Figure~\ref{fig:structure_diagram}.\\

\begin{figure*}[!h]
\centering
\begin{tikzpicture}

\node[text width=3cm] at (-6.2, 0.22) {\bfseries \large {\bfseries \textcolor{red}{\textbf{Inachievability}}} \\ \textcolor{red}{\textbf{results}}};

\node[text width=3cm] at (-6.2,-1.25) {\bfseries \large {\bfseries \textcolor{red}{\textbf{Achievability}}} \\ \textcolor{red}{\textbf{results}}};

\node (upnode1) [draw=black, dashed, thick, rounded corners, minimum width= 15cm, minimum height = 1.4cm, align=center] at (-0.35, 0.05){ };


\node (upnode2) [draw=black, dashed, thick, rounded corners, minimum width= 15cm, minimum height = 2.5cm, align=center] at (-0.35, -1.9){ };


\node (mnode1) [draw=black, rounded corners, fill=chromeyellow, minimum width=0.5cm, align=center] at (-3, 0) {\large Exact Recovery \\ Theorem~\ref{thm:statis_lb}};

\node (mnode2) [draw=black, rounded corners, fill=chromeyellow, minimum width=0.5cm, align=center] at (3, 0) {\large Approximate Recovery \\ Theorem~\ref{thm:statis_approx_lb}};

\node (mnode3) [draw=black, rounded corners, fill=chromeyellow, minimum width=0.5cm, align=center] at (-3, -1) {\large ML Estimator};

\node (mnode4) [draw=black, rounded corners, fill=chromeyellow, minimum width=0.5cm, align=center] at (3, -1) {\large Practical Estimator};

\node (basenode1) [draw=black, rounded corners, fill=columbiablue, minimum width=0.5cm, align=center] at (-6, -2.5) {Noiseless Case \\ Theorem~\ref{thm:noiseless}};

\node (basenode2) [draw=black, rounded corners, fill=columbiablue, minimum width=0.5cm, align=center] at (-3, -2.5) {Noisy Case \\ Theorem~\ref{thm:noisy_ml}};


\node (basenode4) [draw=black, rounded corners, fill=columbiablue, minimum width=0.5cm, align=center] at (1.25, -2.5) {\small Permutation Recovery \\ {Theorem~\ref{thm:robust_lasso_permutation}\footnotesize}};

\node (basenode5) [draw=black, rounded corners, fill=columbiablue, minimum width=0.5cm, align=center] at (5.25, -2.5) {\small Support Set Detection \\ {Corollary~\ref{corol:robust_lasso_support_set} \footnotesize}};


\draw [->, line width=0.5mm] (mnode3) -- (basenode1);
\draw [->, line width=0.5mm] (mnode3) -- (basenode2);

\draw [->, line width=0.5mm] (mnode4) -- (basenode4);
\draw [->, line width=0.5mm] (mnode4) -- (basenode5);
\end{tikzpicture}
\caption{A diagram illustration for the roadmap of the main results to be presented in this paper.\\ \textbf{Upper panel:} inachievability results; \textbf{Lower panel:} achievability results.}
\label{fig:structure_diagram}
\end{figure*}
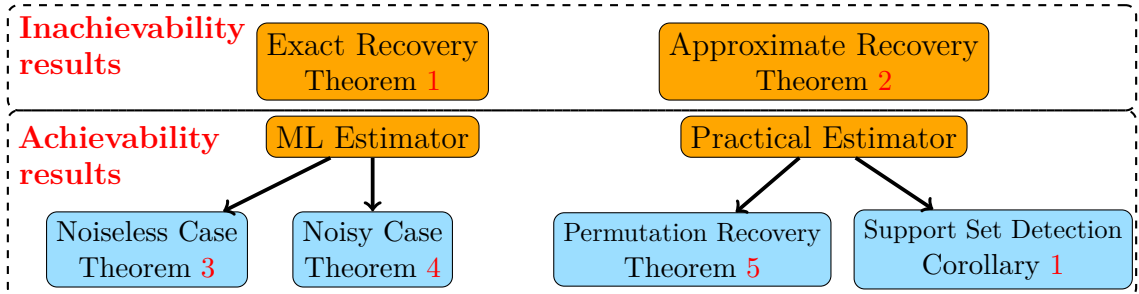

\section{Statistical Lower Bounds} \label{sec:statis_lb}

This paper focuses on recovering both the permutation matrix
$\bPitrue$ and the support set $\supp(\bbetatrue)$,
which are affected by both the sample number $n$ and
$\snr$.
In this section, we will separately discuss
their roles and establish the corresponding statistical lower bounds.

\subsection{Lower bound on sample number $n$}
\label{subsec:n_lb}

To free us from the impact of $\snr$,
we consider the oracle scenario with $\bPitrue$ being known
and limit ourselves to the small noise case, i.e., $\bw \approxeq \bZero$.
Then, our problem reduces to the
classical setting of CS~\citep{donoho2006compressed, candes2006stable, candes2006robust},
where $n \gsim k\log p$ is required for the Gaussian $\bX$ (cf. P.~$507$, Example $15.18$ in~\citet{wainwright_2019}).
This order applies to our case as well since it is hopeless
to recover $\supp(\bbetatrue)$ provided we fail even in the
oracle scenario with known $\bPitrue$.

 \subsection{Lower bound on $\snr$}

This subsection studies the mini-max lower bound
for the $\snr$. The main result is the following.
\begin{theorem}
\label{thm:statis_lb}
We have
\begin{align}
\label{eq:snr_statis_minimax}
\inf_{\wh{\bPi}, \wh{\bbeta}}
\sup_{\bPitrue, \bbetatrue}
\Expc_{\bX,\bw}\Ind\Bracket{(\bPitrue, \supp(\bbetatrue)) \neq (\wh{\bPi}, \supp(\wh{\bbeta}))}\geq \frac{1}{2},
\end{align}
if $n\log\bracket{1+ \snr } + 2 \leq \log\bracket{\abs{\pset_n} {p\choose k}}$,
where $\Expc_{\bX,\bw}(\cdot)$ is taken w.r.t $\bX$ and
$\bw$, and the infimum is over all possible
estimators $\wh{\bPi}$ and $\wh{\bbeta}$.
\end{theorem}

\par \noindent
To better understand Theorem~\ref{thm:statis_lb}, we
spell out the constants and only focus on the orders. Without
any prior information about $\bPitrue$, we can assume
it to distribute uniformly among the set $\pset_n$.
Then, we have $\abs{\pset_n} = \log n!$ and can rewrite
the $\snr$ requirement in \eqref{eq:snr_statis_minimax} as
\begin{align}
\log\bracket{1 + \snr} \lsim \log n + \frac{k\log\bracket{{p}/{k}}}{n}.
\label{eq:statis_special_case}
\end{align}
Compared with~\citet{pananjady2016linear} which requires
$\log(\snr) \asymp \log n$ and $n = \Omega(p)$ for
correct permutation recovery, our bound
only has a slight increase of
$\snr$ requirement in \eqref{eq:statis_special_case}
since $n\gsim k\log p$.
Such an increase of required signal length
is outweighed by the significant
reduction of sample number, which
is from $\Omega(p)$ to $\Omega(k\log p)$.
In addition, we believe that
this theorem can be safely relaxed to
the scenario where $\bX_{ij}$ is an i.i.d. sub-gaussian random variable
with zero mean and unit variance.

The rigorous proof of Theorem~\ref{thm:statis_lb} is  in
Section~\ref{subsec:code_theory_explan}.
Here we only present an intuitive explanation,
which comes from \emph{coding theory}
~\citep{cover2012elements}. The basic idea is to recast the
problem of recovering
$(\bPitrue, \supp(\bbetatrue))$
as a decoding problem~\citep{pananjady2016linear, zhang2019permutation}.
First, we encode $(\bPitrue, \supp(\bbetatrue))$
into the codeword $\bPitrue \bX\bbetatrue$.
Then, we pass it through the additive Gaussian channel~\citep{cover2012elements} and
observe $\by = \bPitrue\bX\bbetatrue + \bw$.
Our goal is to decode $(\bPitrue, \supp(\bbetatrue))$
from the received signal $\by$. An illustration is available in
Figure~\ref{fig:code_trans}.

\begin{figure}[h]
\centering
\includegraphics[width = 4in]{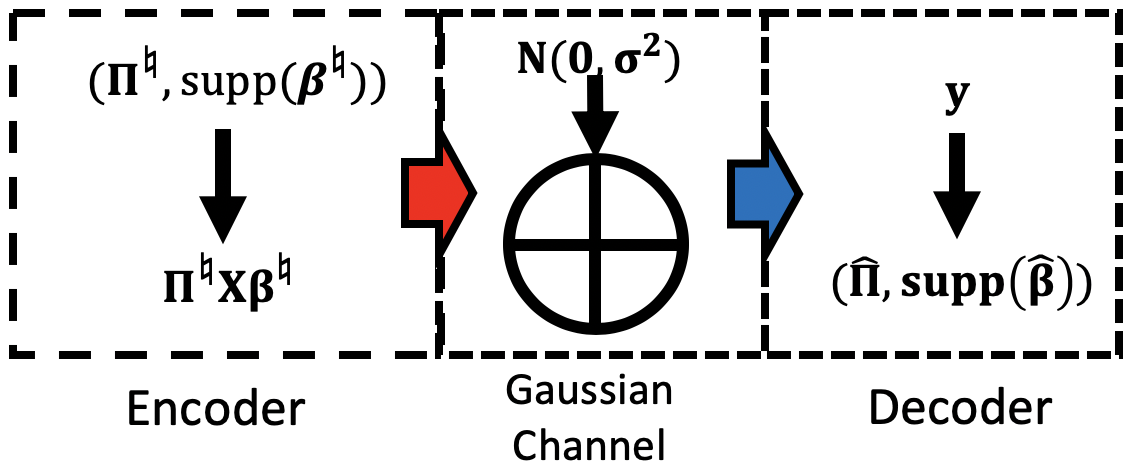}

\vspace{-0.1in}

\caption{Interpretation of Theorem~\ref{thm:statis_lb} from the viewpoint of coding theory.}
\label{fig:code_trans}
\end{figure}

\newpage

Different from~\citet{pananjady2016linear},
we cannot assume $\bbetatrue$ to be given a prior
as this will lead to
absence of sparsity number in the $\snr$ requirement.
Instead, we
assume $\bbetatrue$ to be a binary vector, namely, $\bbetatrue\in \set{0, 1}^n$.
A bonus of this assumption is that
$\supp(\bbetatrue)$ contains the same amount of
information as $\bbeta$, in other words,
there is no extra information (e.g., the specific
values of the non-zero entries) required for
encoding $(\bPitrue,\supp(\bbetatrue))$ into
$\bPitrue \bX\bbetatrue$.
On one hand, the code rate $\mathsf{Rate}$ is computed as
\[
\mathsf{Rate} \defequal \frac{\log\bracket{{p \choose k} n!}}{n} \gsim \
\bracket{1 + \nfrac{1}{2n}}\log n+ \
\frac{k\log(p/k)}{n}.
\]
Meanwhile, the channel capacity is approximately
\[
\mathsf{Capacity}\defequal\dfrac{1}{2}\log\bracket{1 + \frac{\|\bPitrue \bX \bbetatrue\|_{2}^2}{n\sigma^2}}
\approx \dfrac{1}{2}\log\bracket{1 + \frac{\|\bbetatrue\|_{2}^2}{\sigma^2}}.
\]

According to~\citet{cover2012elements}, we require
the code rate is no greater than the channel capacity, i.e.,
$\mathsf{Rate} < \mathsf{Capacity}$, to ensure correct recovery, which naturally yields the
$\snr$ requirement in \eqref{eq:snr_statis_minimax}.

In addition, we notice that the exact recovery
 may be unnecessary in certain
applications.
Following a similar approach, we obtain an
analogous lower bound for the approximate recovery, namely,
$\dH(\wh{\bPi}, \bPitrue) +\dH(\supp(\bbetatrue), \supp(\wh{\bbeta})) \geq \mathsf{D}$, where $\mathsf{D} \geq 0$ is some positive integer. A formal statement is given as follows.

\begin{theorem}
\label{thm:statis_approx_lb}
Provided that $n\log(1 + \snr) +\log 4 \leq \log \zeta$,
we conclude
\[
\inf_{\wh{\bPi}, \wh{\bbeta}}
\sup_{\bPitrue, \bbetatrue}\Expc_{\bX, \bw}\Ind\Bracket{\dH(\wh{\bPi}, \bPitrue) +\dH(\supp(\bbetatrue), \supp(\wh{\bbeta})) \geq \mathsf{D}}\geq \frac{1}{2},
\]
where $\zeta$ is defined as
\begin{align}
\zeta & \defequal
\frac{p!}{(k!)^2[(p-k)!]^2} \cdot  \Bracket{\sum_{i=1}^{\mathsf{D}} \sum_{j=1}^{(\mathsf{D}-i)\vcap k}\
\frac{1}{(n-i)!(k-j)!(p-k-j)!(j!)^2}}^{-1}.
\label{eq:approx_zeta_def}
\end{align}
\end{theorem}

\begin{remark}
Due to the complicated form of $\zeta$ in \eqref{eq:approx_zeta_def},
we only calculate one special case, i.e., $\mathsf{D}= 0$, to illustrate its behavior.
Notice that $\mathsf{D}=0$ corresponds to the exact recovery,
which restore the setting to Theorem~\ref{thm:statis_lb}.
Parameter $\zeta$ under this case is written as
$n!{p\choose k}$, which exhibits
the same order as in Theorem~\ref{thm:statis_lb}.
\end{remark}

In this section, we have established the lower bounds, which remain valid
regardless of the estimators being used.
In the next section,
we will confirm their tightness.

\section{The Maximum Likelihood Estimator}\label{sec:ml_estim}
We will show that the lower bounds in Section~\ref{sec:statis_lb} can be matched
with the differences only up to some
multiplicative constants, to be more specific,
$(i)$ sample number $n$ can be picked as
$n \asymp k\log p$; and $(ii)$ $\snr$ can be set as
$\log\snr\asymp \log n+ \frac{k\log (ep/k)}{n}$.

\subsection{A warm-up example: the noiseless case}
First, we study the role of sample number $n$.
To free it from the impact of $\snr$, we consider the
noiseless case, where $\snr$ is infinite.
Assuming the sparsity number $k$ is given in advance,
we will recover $(\bPitrue, \bbetatrue)$
via the
\emph{maximal likelihood} (ML) estimator reading as
\begin{align}
\label{eq:ml_est}
(\wh{\bPi}_{\mathsf{ML}}, \wh{\bbeta}_{\mathsf{ML}}) = \
\argmin_{\substack{\bPi \in \pset_n \\ \norm{\bbeta}{0}\leq k}}
\norm{\by - \bPi \bX \bbeta}{2}.
\end{align}
Then we  have the following theorem:

\begin{theorem}
\label{thm:noiseless}
Consider the noiseless case where $\bw = \bZero$.
Given that
$n\gsim k\log p$, we have \\
$\Prob((\wh{\bPi}_{\mathsf{ML}}, \wh{\bbeta}_{\mathsf{ML}})\neq (\bPitrue,~\bbetatrue)) \lsim n^{-2}$.
\end{theorem}

\begin{proofoutline}
We divide the analysis of the reconstruction
error $(\wh{\bPi}_{\mathsf{ML}}, \wh{\bbeta}_{\mathsf{ML}})\neq (\bPitrue, \supp(\bbetatrue))$
into three categories:
$(i)$ $\big \{\bPitrue =\wh{\bPi}_{\ml} , \
\bbetatrue \neq \wh{\bbeta}_{\ml} \big \}$;
$(ii)$ $\big \{\bPitrue \neq\wh{\bPi}_{\ml}, \
\bbetatrue =\wh{\bbeta}_{\ml}\big \}$; and
$(iii)$ $\big\{\bPitrue \neq\wh{\bPi}_{\ml}, \
\bbetatrue \neq\wh{\bbeta}_{\ml} \big \}$.
Iterating over all possible pairs
$(\bPi,~\supp(\bbeta))$,
we will show the above $3$ types of errors rarely happen.
The detailed proof is deferred to Subsection~\ref{thm_proof:noiseless}.
\end{proofoutline}
We notice that Theorem~\ref{thm:noiseless} directly recovers
a sparse signal $\bbetatrue$, which contains more information
than its support set $\supp(\bbetatrue)$.
Moreover, we observe that the sample number requirement
in Theorem~\ref{thm:noiseless} is well aligned with the
lower bound presented in Subsection~\ref{subsec:n_lb} and
can hence confirm its tightness.

\subsection{The noisy case}

This subsection investigates the noisy case.
To correctly recover the support set $\supp(\bbetatrue)$,
we need an additional assumption on non-zero entries' smallest magnitudes.
Otherwise, even in the classical setting without any permutation,
small sensing noise can lead
to incorrect support set detection.
The formal statement reads as follows.

\begin{theorem}
\label{thm:noisy_ml}
Provided that
\begin{enumerate}
    \item[(i)] $n \gsim k\log p$,
\item [(ii)] $\log\snr \gsim \log n  + \frac{k}{n}\log(\frac{ep}{k})$
\item [(iii)] $\min_{i\in \supp(\bbetatrue)}\nfrac{|\bbetatrue_i|^2}{\sigma^2}\gsim 1$,
\end{enumerate}
we have
$\Prob((\wh{\bPi}_{\mathsf{ML}}, \supp(\wh{\bbeta}_{\mathsf{ML}})) \neq (\bPitrue, \supp(\bbetatrue)) \lsim n^{-2} + e^{-c k\log p}$.
\end{theorem}
Compared with the result pertaining to the noiseless case (i.e.,
Theorem~\ref{thm:noiseless}), Theorem~\ref{thm:noisy_ml} requires
non-zero entries' magnitudes to be at least some positive
constants. Apart from this constraint, our bound matches the
minimax lower bound in Theorem~\ref{thm:statis_lb} up to some
multiplicative constants, specifically,
$n\gsim k\log p$ and $\log \snr \gsim \log n + \nfrac{k\log p}{n}$.

\begin{remark}
With a slight modification of the ML estimator,
we can significantly relax
the third assumption in the above theorem, i.e., $\min_{i\in \supp(\bbetatrue)}\nfrac{|\bbetatrue_i|^2}{\sigma^2}\gsim 1$.
Notice that the correct permutation recovery only requires the first two assumptions in Theorem~\ref{thm:noisy_ml}.
Hence, we can first reconstruct the permutation matrix
$\bPitrue$ with the ML estimator. Afterwards, we restore the shuffled sparse recovery problem to its classical setting and invoke the
previous works to detect the support set $\supp(\bbetatrue)$. With
the above modifications, we can improve the assumption $\min_{i\in \supp(\bbetatrue)}\nfrac{|\bbetatrue_i|^2}{\sigma^2}\gsim 1$
to  $\min_{i\in \supp(\bbetatrue)}\nfrac{|\bbetatrue_i|^2}{\sigma^2}\gsim \nfrac{\log p}{n}$.
\end{remark}

We would like to emphasize that the ML estimator
only serves to confirm the tightness of Theorem~\ref{thm:statis_lb}.
It is unpractical due to its high computational cost:
 it needs to iterate
$(i)$ all possible $k$-sparse subsets, which consists of
${p\choose k}$ cases; and $(ii)$
all possible permutation matrices $\bPi$, which consists $n!$ cases.
The next section will present a computational-friendly estimator.

\section{Practical Estimator}\label{sec:prac_estim}

This section proposes a practical estimator to combat the high computational
cost associated with the ML estimator,
which consists of
two stages: permutation recovery and
support set detection.
A formal statement is in
Algorithm~\ref{alg:prac_estim}.

\subsection{Permutation Recovery}
We note that the major difficulty in the permutation recovery
stems from the missing value of signal $\bbetatrue$.
One natural solution is to restore the permutation
with an approximate value of $\bbetatrue$.
To begin with, we impose a
parsimonious assumption on the number of permuted
rows, or equivalently, only a small proportion of rows are permuted. Then, we
adopt a denoising viewpoint and treat $\bracket{\bI-\bPitrue}\bX\bbetatrue$ as a sparse outlier
to be removed. Inspired by~\citet{nguyen2013robust} and~\citet{slawski2017linear}, we can estimate
the signal $\bbetatrue$ as
\[
(\wt{\bXi}, \wt{\bbeta}) = \argmin_{\bXi, \bbeta}
~& \frac{1}{2n}\norm{\by - \bX\bbeta - \sqrt{n}\cdot \bXi}{2}^2 + \lambda_{\bXi} \norm{\bXi}{1}
+ \lambda_{\bbeta} \norm{\bbeta}{1}.
\]
Afterwards, we reconstruct the
permutation matrix $\bPitrue$ via the following
\emph{linear assignment problem} (LAP), which
reads as
\[
\wh{\bPi} = \argmax_{\bPi}
\langle \by,~\bPi \bX\wt{\bbeta}\rangle,
\]
where $\wh{\bbeta}$ is the solution of
\eqref{eq:robust_lasso_optim_def}.
Then, we conclude
\begin{theorem}
\label{thm:robust_lasso_permutation}
We set $\belambda$ and $\xilambda$ in
\eqref{eq:robust_lasso_optim_def} as
$c_0 \sigma \sqrt{\nfrac{\log p}{n}}$
and
$c_1 \sigma \sqrt{\nfrac{\log n}{n}}$,
respectively. Assuming that
$(i)$ $n\gsim k\log p$, $(ii)$ $h \lsim \frac{n}{\log n}$, and $(iii)$
\[
\snr \gsim \frac{n^{2(1+\varepsilon)} (n-1)^2}{4\pi}
\bigg[&\sqrt{\log np} \bracket{k\sqrt{\frac{\log p}{n}}
\vcup h \sqrt{\frac{\log n}{n}}
} + 2\log(n^{1+\varepsilon}(n-1)) \bigg]^2,
\]
we conclude that \eqref{eq:robust_lasso_pi_optim} can yield the ground truth with probability
exceeding $1-2n^{-\varepsilon}$, i.e.,
 $\Prob(\wh{\bPi} = \bPitrue)\geq 1 - 2n^{-\varepsilon}$.
\end{theorem}
First, we discuss the $\snr$ requirement.
From the above theorem, we conclude
that
the correct permutation matrix
can be obtained provided that $\log \snr \gsim \log n$, which matches the
mini-max lower bound in
Theorem~\ref{thm:statis_lb} up to some
multiplicative constant.
Then, we consider the maximum allowed number of permuted rows,
i.e., $h \lsim \frac{n}{\log n}$.
Compared with the optimal order
$O(n)$, we experience a loss of
logarithmic term. This is consistent with our parsimonious assumption on the number of permuted
rows, i.e., $h \ll n$.
Moreover, we discuss the minimum
required sample number $n$, which
is of the order $\Omega(k\log p)$.
Notice that this is the same as
the mini-max bound discussed in
Subsection~\ref{subsec:n_lb}.

\subsection{Support Set Detection}
Once we have the correct permutation matrix $\bPitrue$, we can restore
\eqref{eq:sys_mdl} to the classical model and detect the support set
$\supp(\bbetatrue)$ with the Lasso estimator, which is written as
\[
\wh{\bbeta}_{\lasso}=  \
\argmin_{\bbeta}~\dfrac{1}{2n}\big \|\wh{\bPi}^{\rmt}\by -\bX\bbeta\big \|_{2}^2 + \
\lambda_{\lasso(n)}\norm{\bbeta}{1},
\]
where $\wh{\bPi}$ denotes the solution of
\eqref{eq:robust_lasso_pi_optim}.
Then, we detect the support set $\supp(\bbetatrue)$ by
selecting the entries with the first $k$-largest magnitude.
With the standard results concerning the sign consistency of
Lasso estimator, e.g.,~\citet{lounici2008sup}, we can show the
support set can be detected with high probability under the
settings of Theorem~\ref{thm:robust_lasso_permutation}.

\begin{corollary}
\label{corol:robust_lasso_support_set}	
Under the settings of Theorem~\ref{thm:robust_lasso_permutation},
we pick $\lambda_{\lasso(n)}$ in
\eqref{eq:robust_lasso_beta_support} as
$c \sigma \sqrt{{\log p}/{n}}$.
Provided that $\min_{\bbetatrue_i\neq 0} (|\bbetatrue_i|^2/{\sigma^2}) \gsim \frac{\log p}{n}$,
we have
$\sign(\thres(\wh{\bbeta}; k)) = \sign(\bbetatrue)$ hold with probability $1-o(1)$, where
$\thres(\cdot; k)$ selects the entries with the first $k$-largest magnitude
and is defined in \eqref{eq:hard_thres_def}.
\end{corollary}
This corollary suggests that the support set $\supp(\bbetatrue)$ can be detected with high probability. Compared with Theorem~\ref{thm:robust_lasso_permutation}, Corollary~\ref{corol:robust_lasso_support_set} has one additional
assumption on the smallest magnitude of the non-zero entries in $\bbetatrue$, e.g., $\min_{\bbetatrue_i\neq 0} (|\bbetatrue_i|^2/{\sigma^2}) \gsim \nfrac{\log p}{n}$. Notice that this assumption is quite standard~\citep{lounici2008sup, zhao2006model, wainwright2009sharp} in studying the property of sign consistency.
\begin{remark}
\label{remark:performance_improve}
In~\citet{zhang2021sparse}, we need
\[
\min_{\bbetatrue_i \neq 0}
|\bbetatrue_i| \gsim (1+ k\sqrt{\nfrac{\log p}{n}})\sqrt{\nfrac{\log p}{n}}
 \cdot \bracket{\|\bbetatrue\|_{2}\
\sqrt{h \log n} \vcup \sigma}.
\]
Meanwhile, our estimator improves this requirement to
\[
\min_{\bbetatrue_i \neq 0}
|\bbetatrue_i| \gsim \sigma \sqrt{\frac{ \log p}{n}}.
\]
Compared with~\citet{zhang2021sparse}, our
estimator has a significant improvement. First, our assumption on $\min_{\bbetatrue_i \neq 0}
|\bbetatrue_i|$ is free from the
total energy $\|\bbetatrue\|_2$.
Even after we factor out the impact of
$\|\bbetatrue\|_2$,~\citep{zhang2021sparse}
still requires $\min_{\bbetatrue_i \neq 0}
|\bbetatrue_i| \gsim \frac{\sigma (k \log p)}{n}$ while our estimator
reduces the requirement to
 $\min_{\bbetatrue_i \neq 0}
|\bbetatrue_i| \gsim \sigma \sqrt{\nfrac{ \log p}{n}}$.

\end{remark}

\begin{algorithm}[h]
\caption{Permuted-Lasso Estimator.}
\label{alg:prac_estim}
\begin{algorithmic}[1]
\Statex \textbullet~
\textbf{Input:} observation $\by$, sensing matrix $\bX$, and sparsity number $k$.

\Statex \textbullet~
\textbf{Stage I: Permutation Recovery.}
We pick $\belambda$ and $\xilambda$ as $c_0\sigma \sqrt{\nfrac{\log p}{n}}$
and $c_1 \sigma \sqrt{\nfrac{\log n}{n}}$. We restore
the correspondence information as
\begin{align}
(\wt{\bXi}, \wt{\bbeta}) =~& \argmin_{\bXi, \bbeta}
\frac{1}{2n}\norm{\by - \bX\bbeta - \sqrt{n}\cdot \bXi}{2}^2 + \lambda_{\bXi} \norm{\bXi}{1}
+ \lambda_{\bbeta} \norm{\bbeta}{1}; \label{eq:robust_lasso_optim_def}\\
\wh{\bPi} =~& \argmax_{\bPi}
\langle \by,\bPi \bX\wt{\bbeta}\rangle.
\label{eq:robust_lasso_pi_optim}
\end{align}

\Statex \textbullet~
\textbf{Stage II: Support Set Detection.}
\noindent
With the permutation matrix $\wh{\bPi}$ in
\eqref{eq:robust_lasso_pi_optim},
we pick $\lambda_{\lasso(n)}$ in
\eqref{eq:robust_lasso_beta_support} as
$c_3 \sigma \sqrt{{\log p}/{n}}$ and
detect the support set by first $(i)$ computing
$\wh{\bbeta}_{\lasso}$ as
\begin{align}
\label{eq:robust_lasso_beta_support}
\wh{\bbeta}_{\lasso} =  \
\argmin_{\bbeta}~\dfrac{1}{2n}\big \|\wh{\bPi}^{\rmt}\by -\bX\bbeta\big \|_{2}^2 + \
\lambda_{\lasso(n)}\norm{\bbeta}{1},
\end{align}
and then $(ii)$ performing hard-thresholding to $\wh{\bbeta}_{\lasso}$,
which is
\begin{align}
\label{eq:hard_thres_def}
(\hardthres(\wh{\bbeta}_{\lasso}; k))_i  \defequal \left\{
\begin{aligned}
& (\wh{\bbeta}_{\lasso})_i,~\textup{if}~|(\wh{\bbeta}_{\lasso})_i|\textup{ is among the } k \textup{-largest absolute entries}; \\
&0,~\textup{otherwise}.	
\end{aligned}\right.
\end{align}

\Statex \textbullet~
\textbf{Output:} we return
$(\wh{\bPi}, \hardthres(\wh{\bbeta}_{\lasso}; k))$.
\end{algorithmic}
\end{algorithm}

In the end, we will briefly discuss the potential
methods of recovering
$(\bPitrue, \supp(\bbetatrue))$. Notice that Algorithm~\ref{alg:prac_estim}
only consists of one step of permutation recovery and support
set detection. One natural way for the performance improvement is to iteratively perform the permutation recovery and the support set detection. In addition, we find that $\wt{\bXi}$ in
\eqref{eq:robust_lasso_optim_def} is largely ignored.
Since it contains information about $(\bI - \bPitrue)\bX\bbetatrue$, in other words, it has information
about the permutation matrix, we can use it to refine
the reconstructed permutation.

\newpage

\section{Simulations}\label{sec:simul}

This section presents the numerical results,
where the permutation matrix $\bPitrue$ and the support
set $\supp(\bbetatrue)$ are reconstructed via
Algorithm~\ref{alg:prac_estim}.
The regularizer
coefficients, i.e., $\belambda, \xilambda$, and
$\lambda_{\lasso(n)}$, are all picked as $2.0$.
First, we consider the Gaussian setting, where each
entry $\bX_{ij}$ are i.i.d. standard normal random variables,
i.e., $\normdist(0, 1)$. Moreover, we extend it to the
setting of sub-gaussian distributions, where
$\bX_{ij}$ are i.i.d. sub-gaussian random variables, to be more specific, $\bX_{ij}$ are uniformly distributed within the region $[-1, 1]$, namely, $\bX_{ij}\iid \Unif[-1, 1]$.
\par
We evaluate the performance
in terms of the ratio $\nfrac{\log \snr}{\log n}$, which is widely used in the study of permuted linear regression.
We only plot the correct rate for the permutation recovery,
since the support set detection in \eqref{eq:robust_lasso_beta_support} and \eqref{eq:hard_thres_def} seldom makes any mistake, even with a wrong permutation matrix $\wh{\bPi}$ returned in
\eqref{eq:robust_lasso_pi_optim}.

\subsection{Impact of sparsity number}

This subsection studies the impact of sparsity number $k$. We
fix the signal length $p$ and the permuted row number
$h$ as $500$ and $20$, respectively.
We let the sample number $n\in \set{180, 200, 220}$ and
vary the sparsity number
$k$ within the set $\set{5, 10, 20}$. The numerical
results are put in Figure~\ref{fig:k_impact}.

\vspace{0.1in}\noindent
\textbf{Discussion.}
First, we discuss the Gaussian setting.
From the curves in Figure~\ref{fig:k_impact},
we confirm the correctness of Theorem~\ref{thm:robust_lasso_permutation}, which
suggests that the correct permutation matrix can be obtained once $\log \snr \gsim \log n$.
In addition, we notice that the correct permutation reconstruction
requires a larger $\snr$ with an increasing
sparsity number $k$. For example, we can obtain the ground-truth permutation
matrix with $\log \snr = 5.5\log n$
when $(n, p, h, k)=(180, 500, 20, 5)$. When the sparsity number $k$ increases to $20$, the requirement
for the correct permutation recovery increases to
$\log \snr > 6 \log n$. Similar phenomena can be observed
for other settings as well. Second, we discuss the
uniform distribution setting. Numerical results show
a similar behavior as that of the Gaussian setting
and suggest that our estimator in
Algorithm~\ref{alg:prac_estim} can work beyond the
setting in Theorem~\ref{thm:robust_lasso_permutation}.

\begin{figure}[t!]

\centering

\mbox{\hspace{-0.15in}
\includegraphics[width = 2.75in]{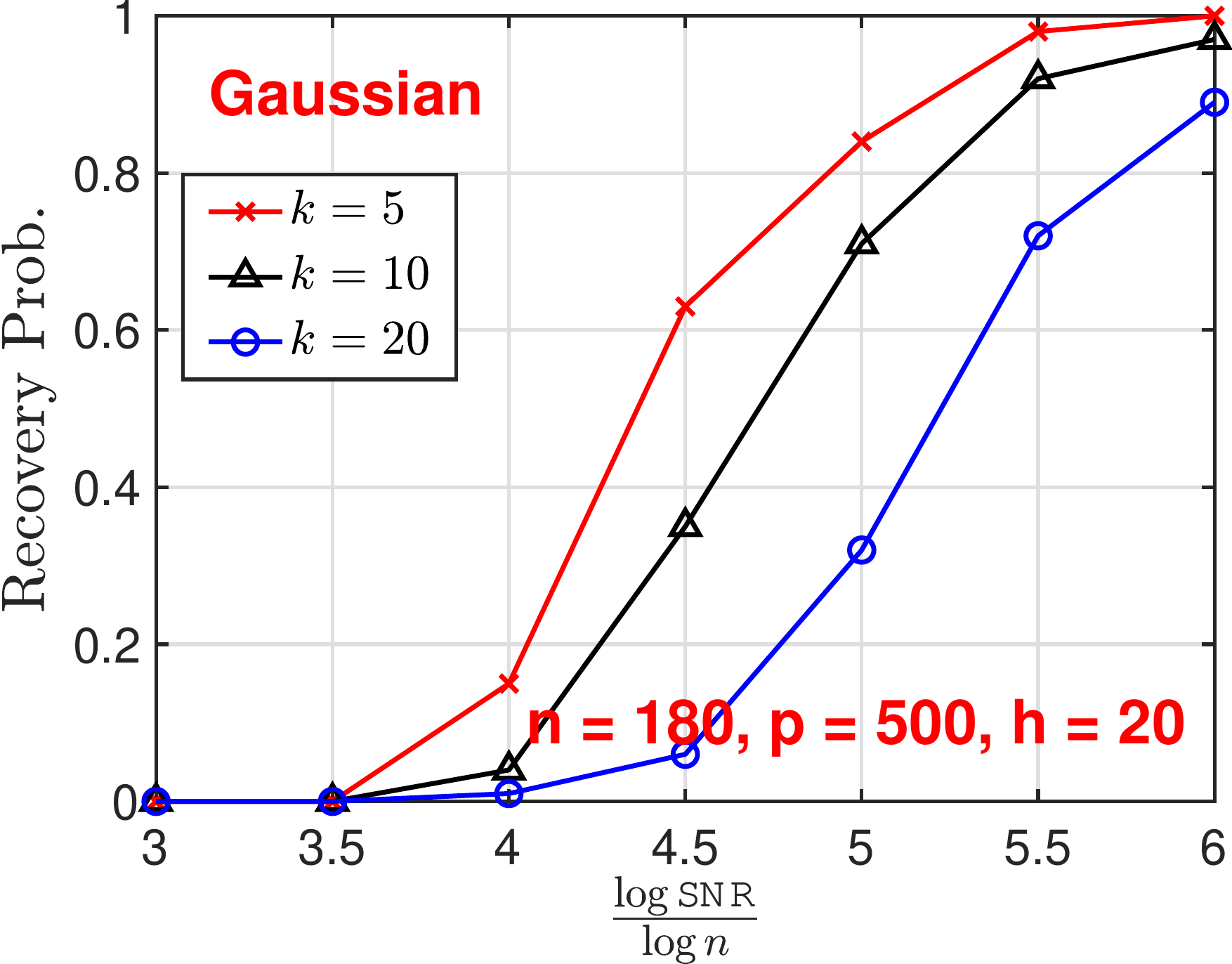}
\includegraphics[width = 2.75in]{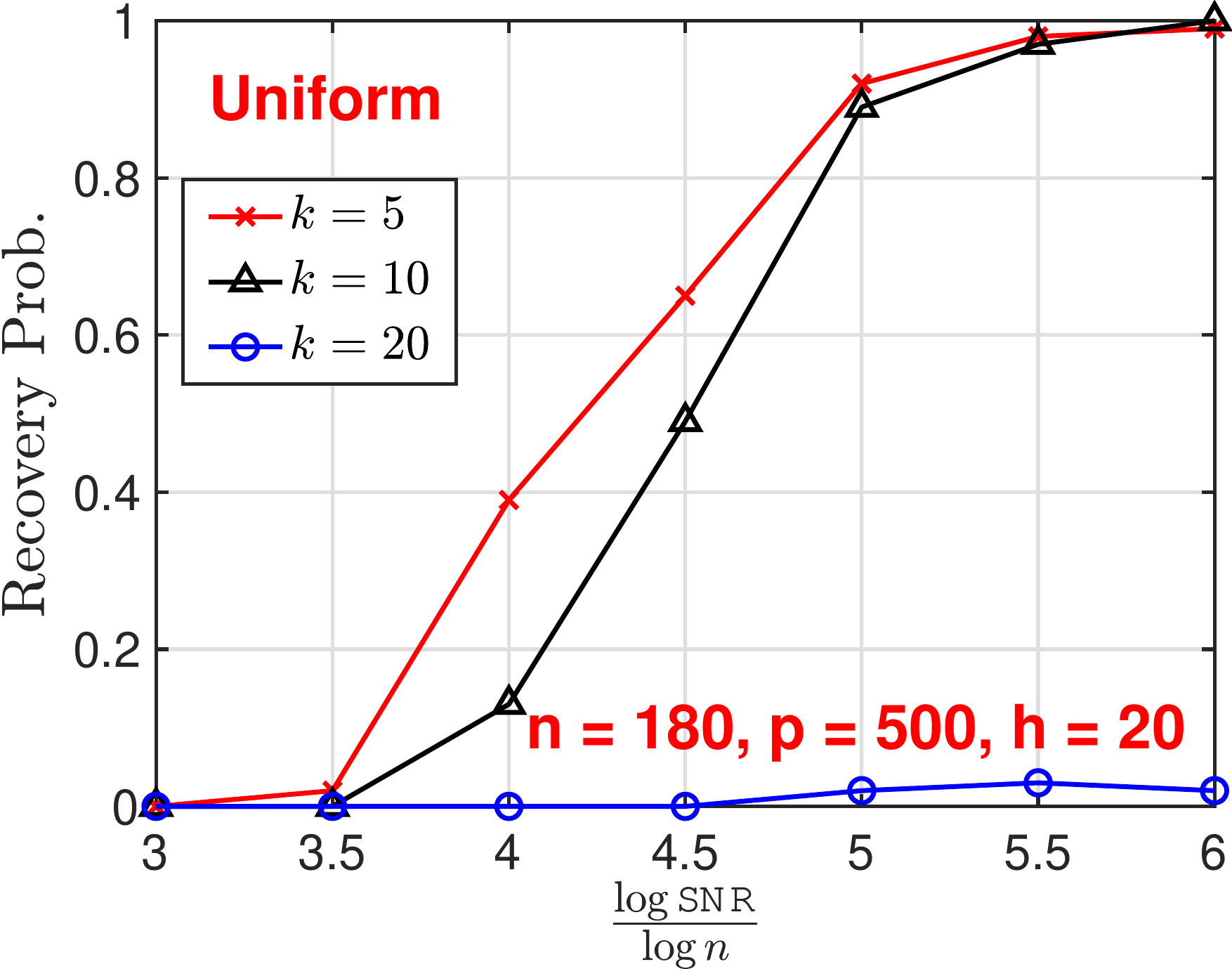}
}

\mbox{\hspace{-0.15in}

\includegraphics[width = 2.75in]{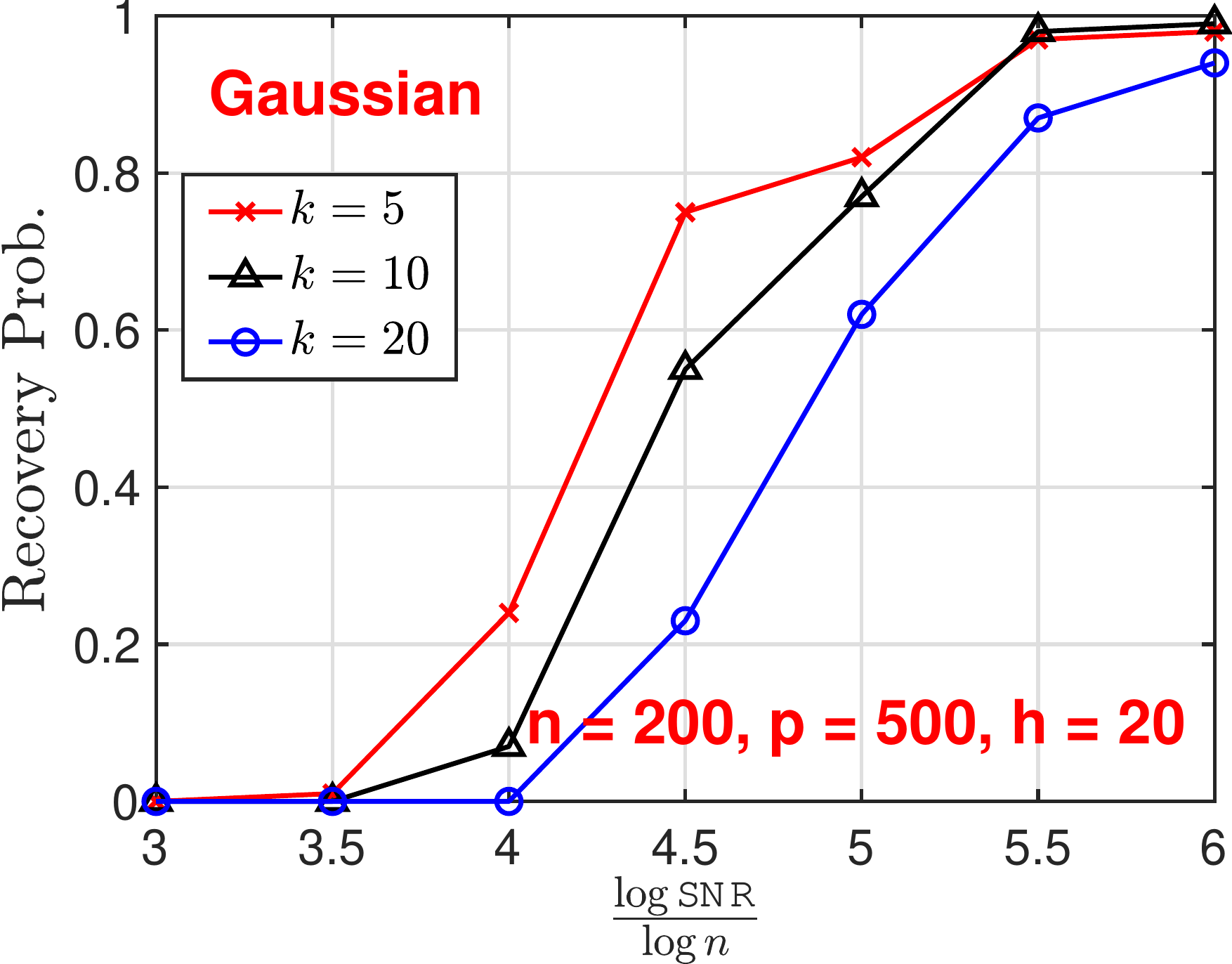}
\includegraphics[width = 2.75in]{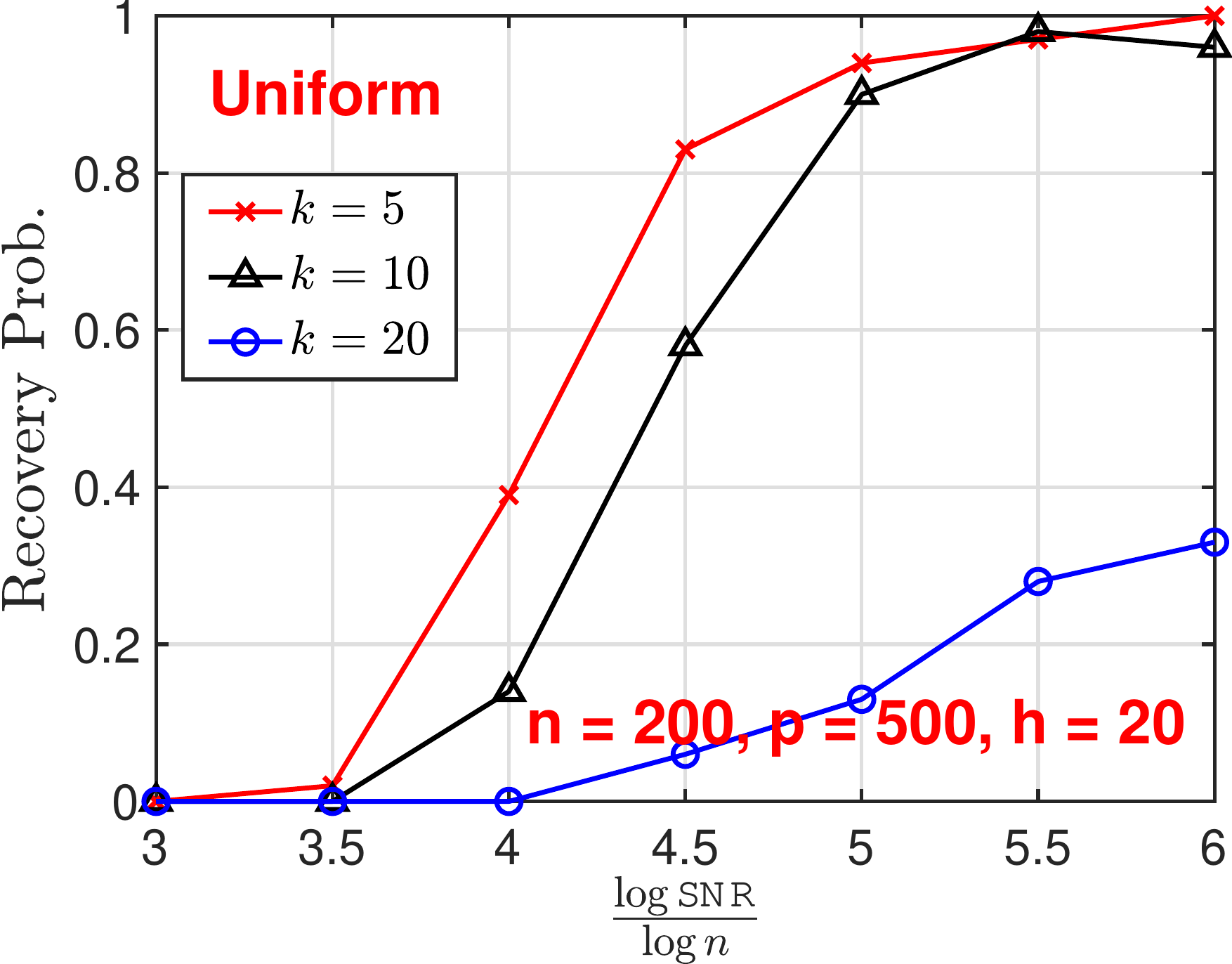}
}

\mbox{\hspace{-0.15in}
\includegraphics[width = 2.75in]{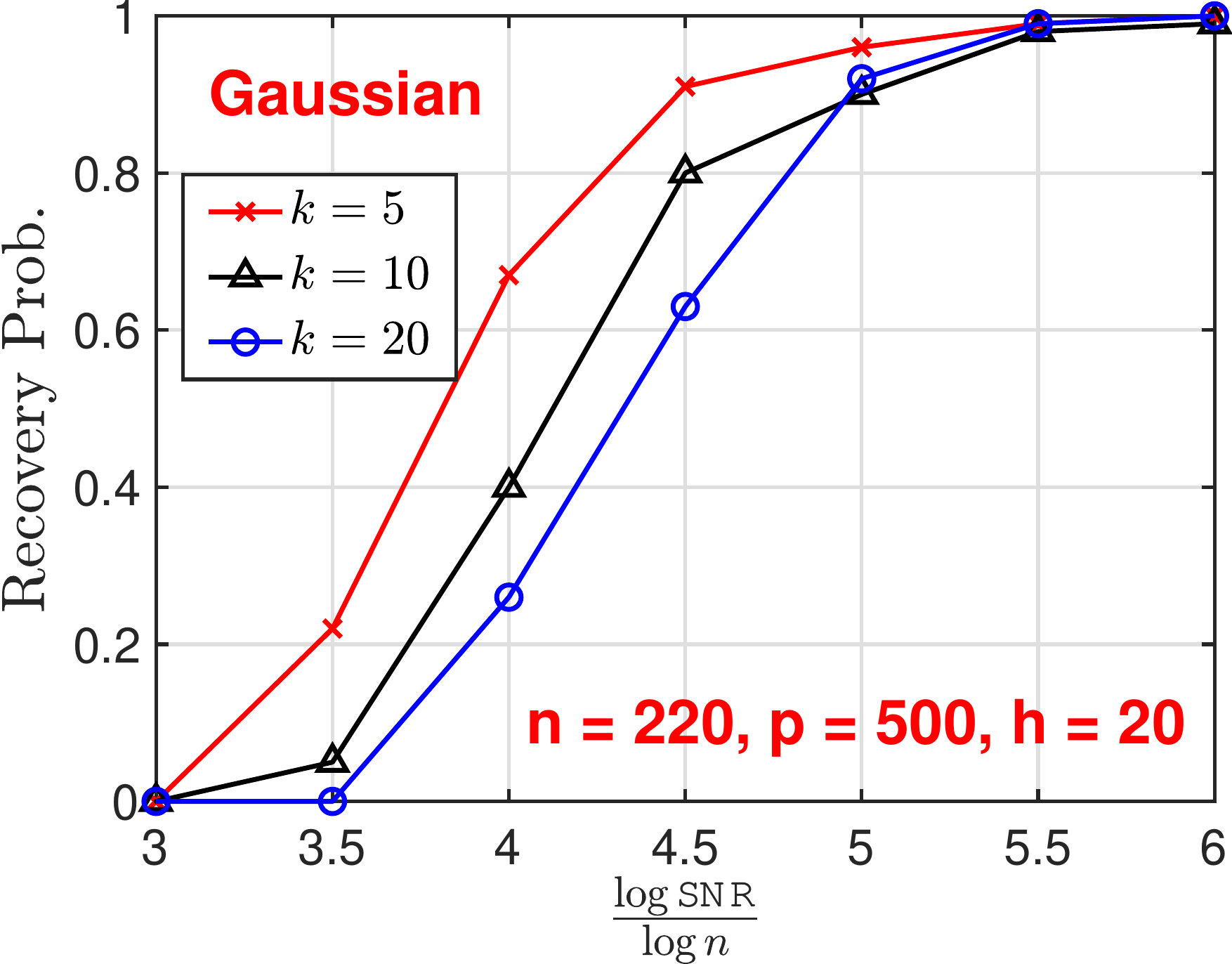}

\includegraphics[width = 2.75in]{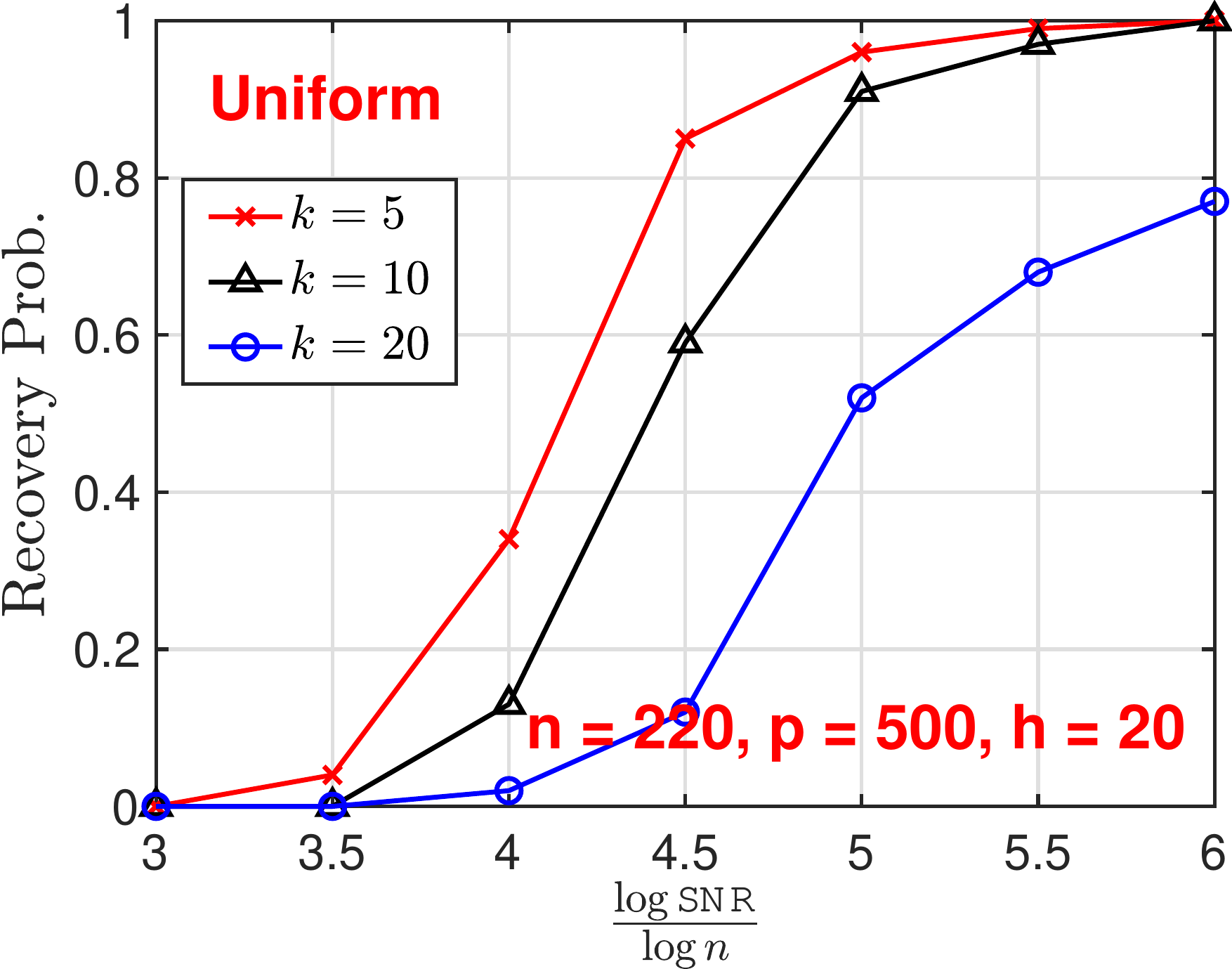}
}
\vspace{-0.05in}

\caption{Simulated permutation recovery rate $\Prob(\wh{\bPi} = \bPitrue)$ with
$n = \{180, 200, 220\}$, $p = 500$, $h = 20$, and $k = \{5, 10, 20\}$, w.r.t. $\nfrac{\log \snr}{\log n}$.
(\textbf{Left Panel}) We have $\bX_{ij}$ be i.i.d. normal random variables, i.e., $\bX_{ij}\iid \normdist(0, 1)$;
(\textbf{Right Panel}) We have $\bX_{ij}$ be i.i.d. sub-gaussian random variables, to be more specific, $\bX_{ij} \iid \Unif[-1, 1]$.
}
\label{fig:k_impact}\vspace{0.2in}
\end{figure}
\clearpage

\begin{figure}[!ht]
\centering

\mbox{\hspace{-0.15in}
\includegraphics[width = 2.75in]{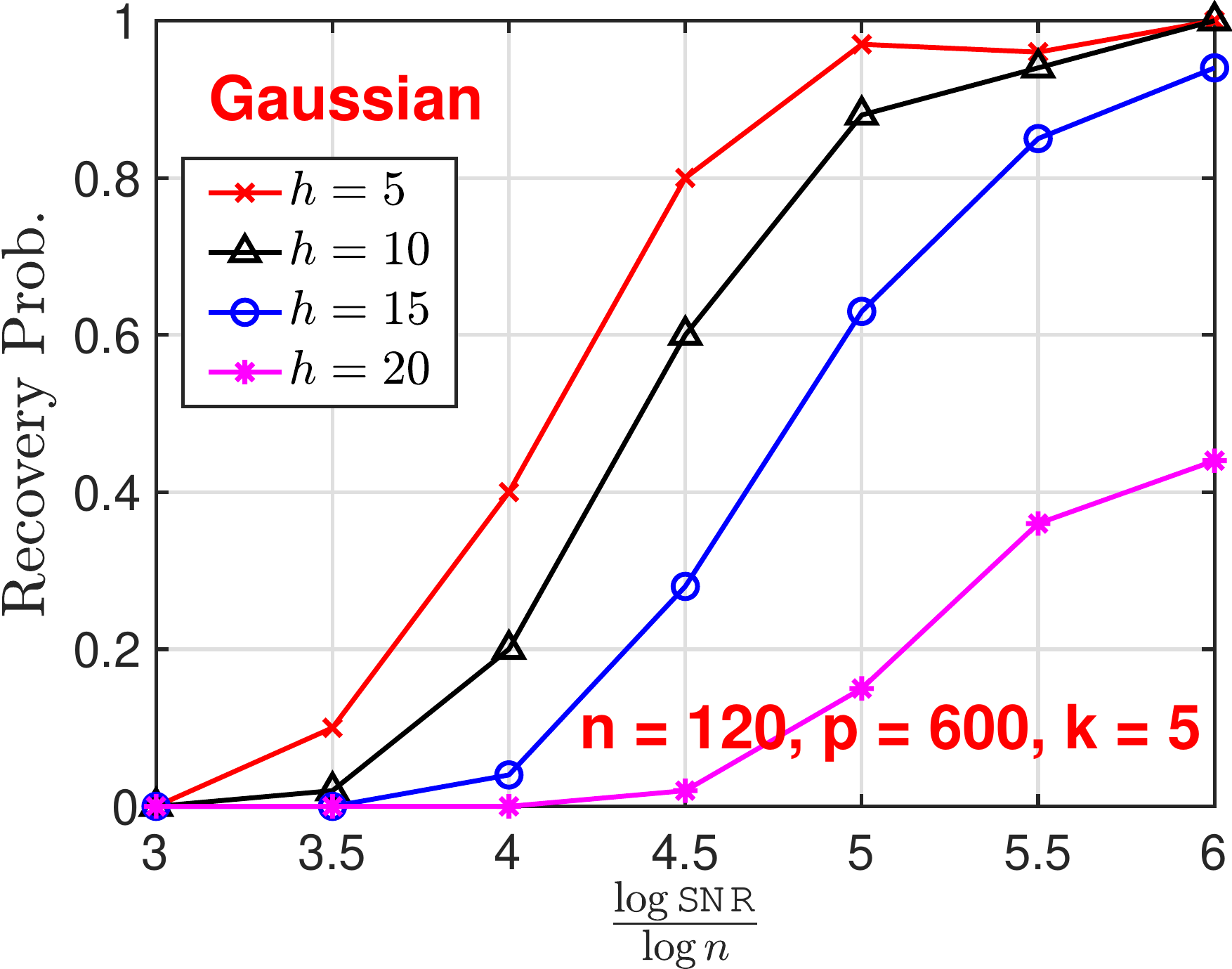}
\includegraphics[width = 2.75in]{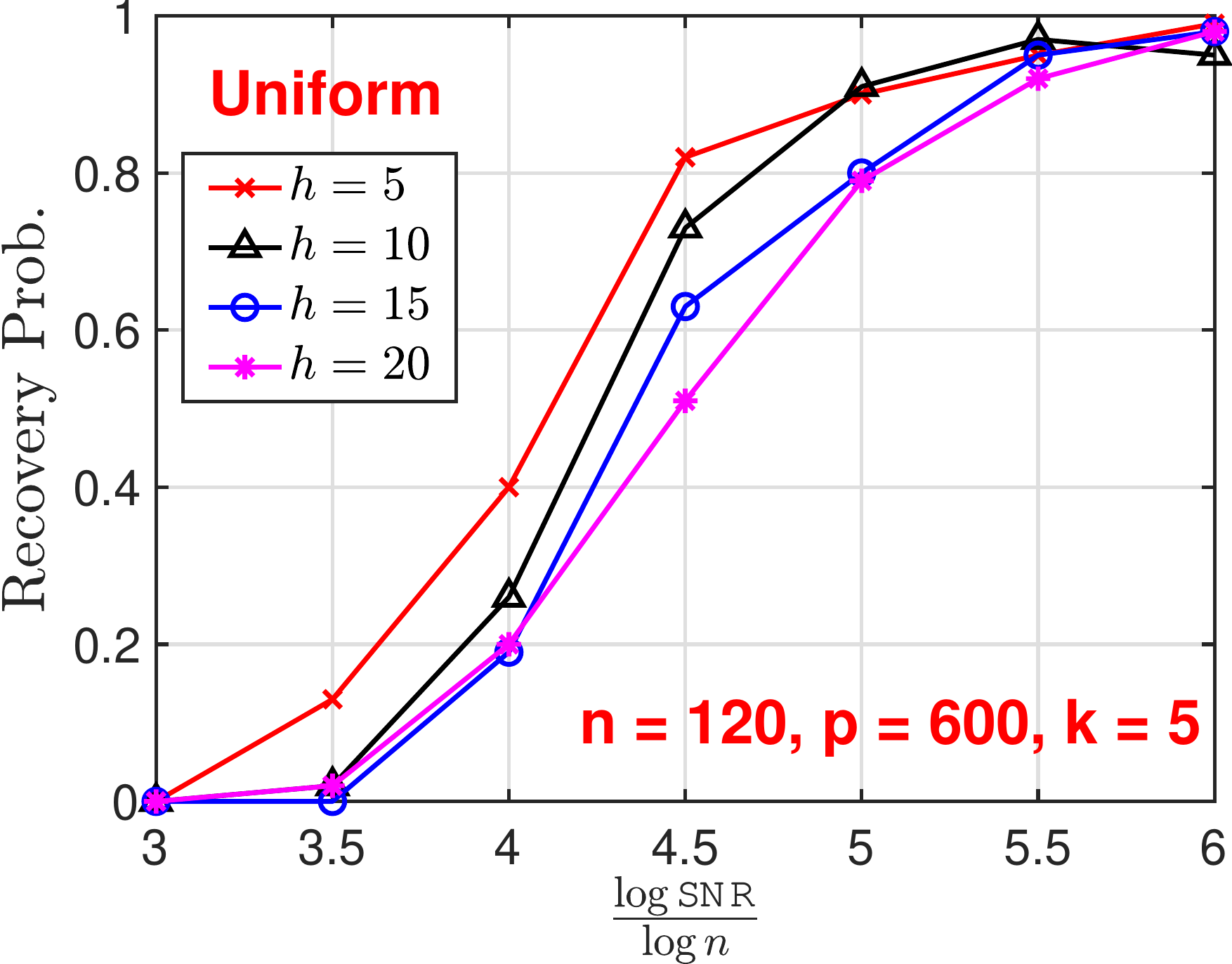}
}

\mbox{\hspace{-0.15in}
\includegraphics[width = 2.75in]{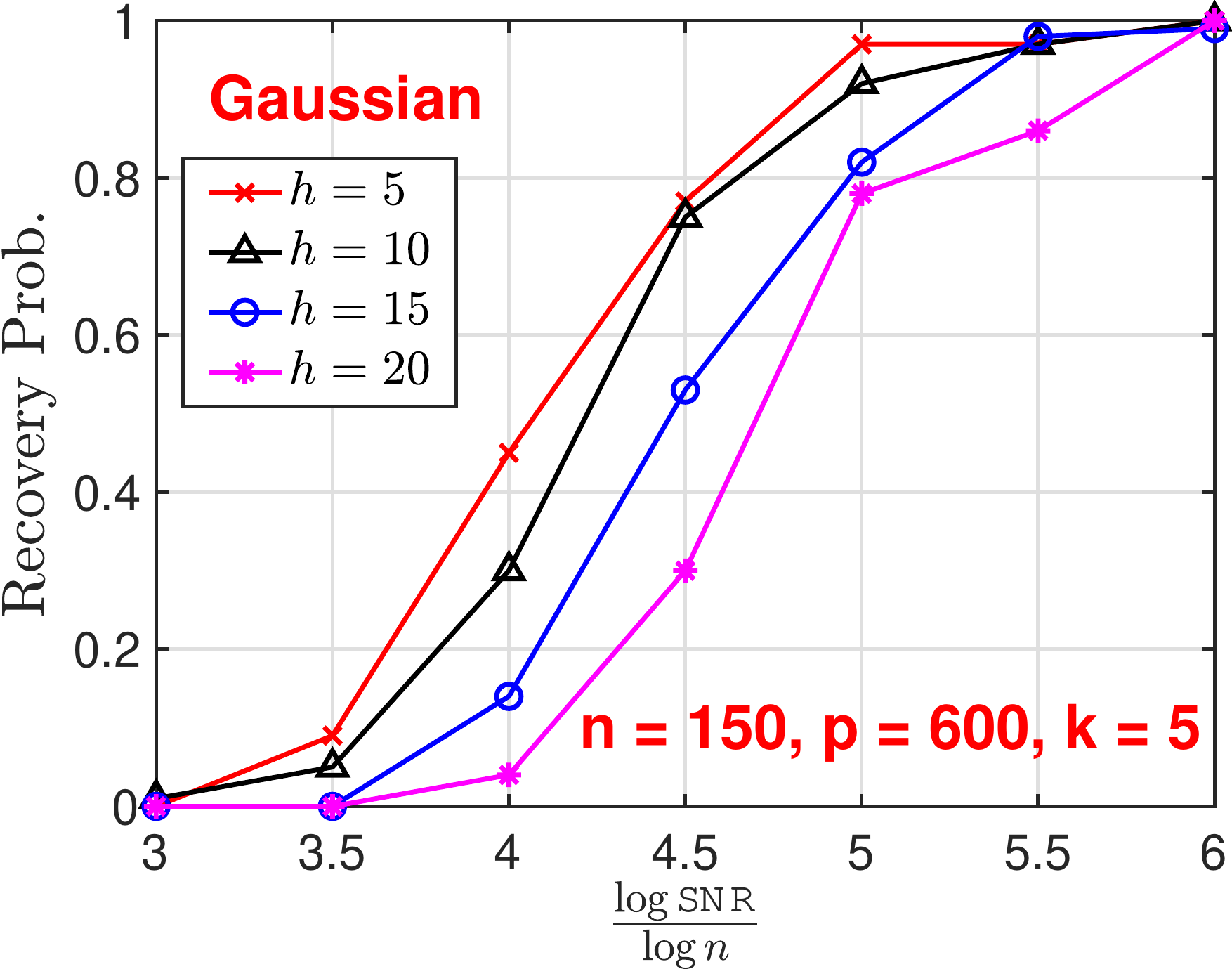}
\includegraphics[width = 2.75in]{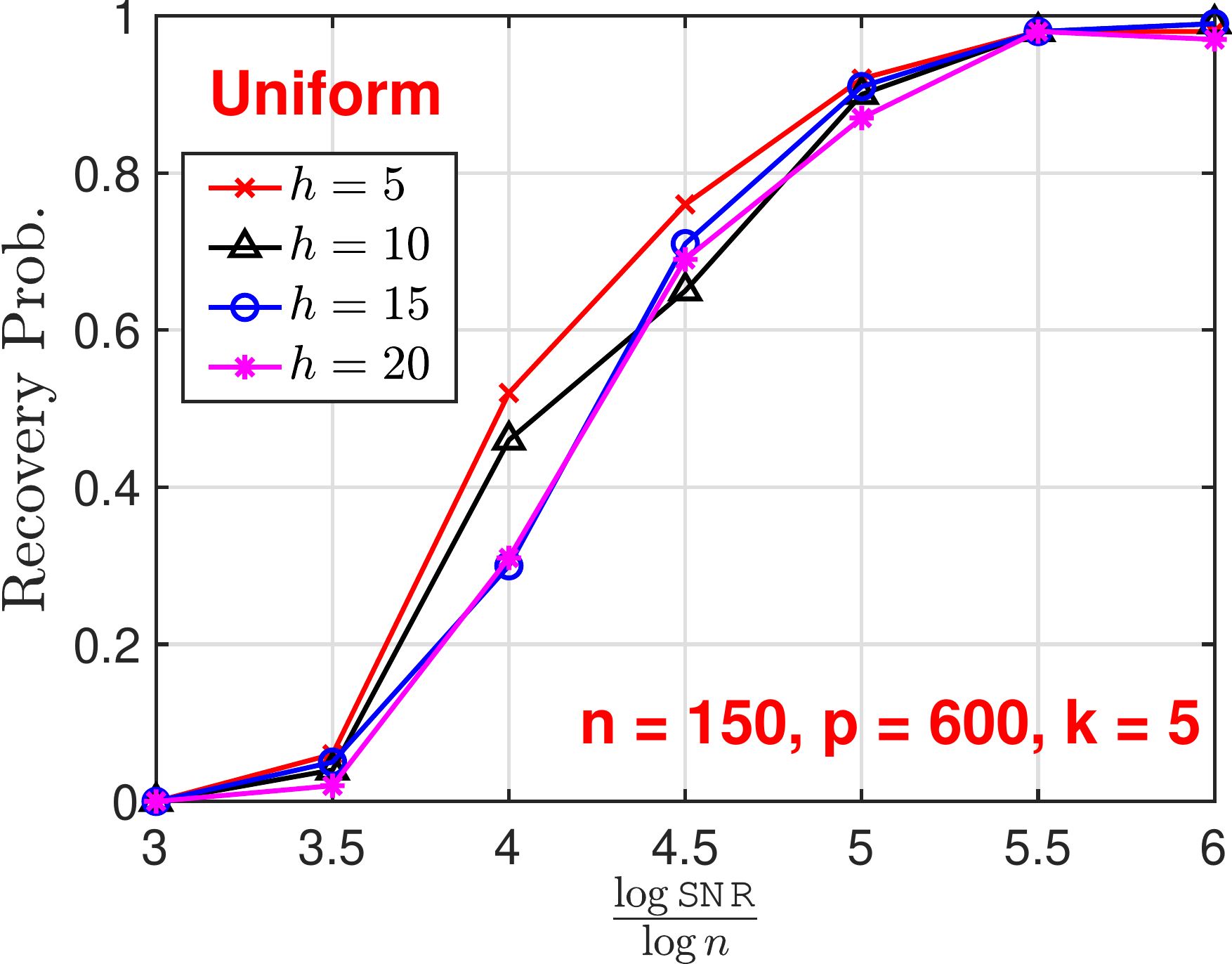}

}

\mbox{\hspace{-0.15in}
\includegraphics[width = 2.75in]{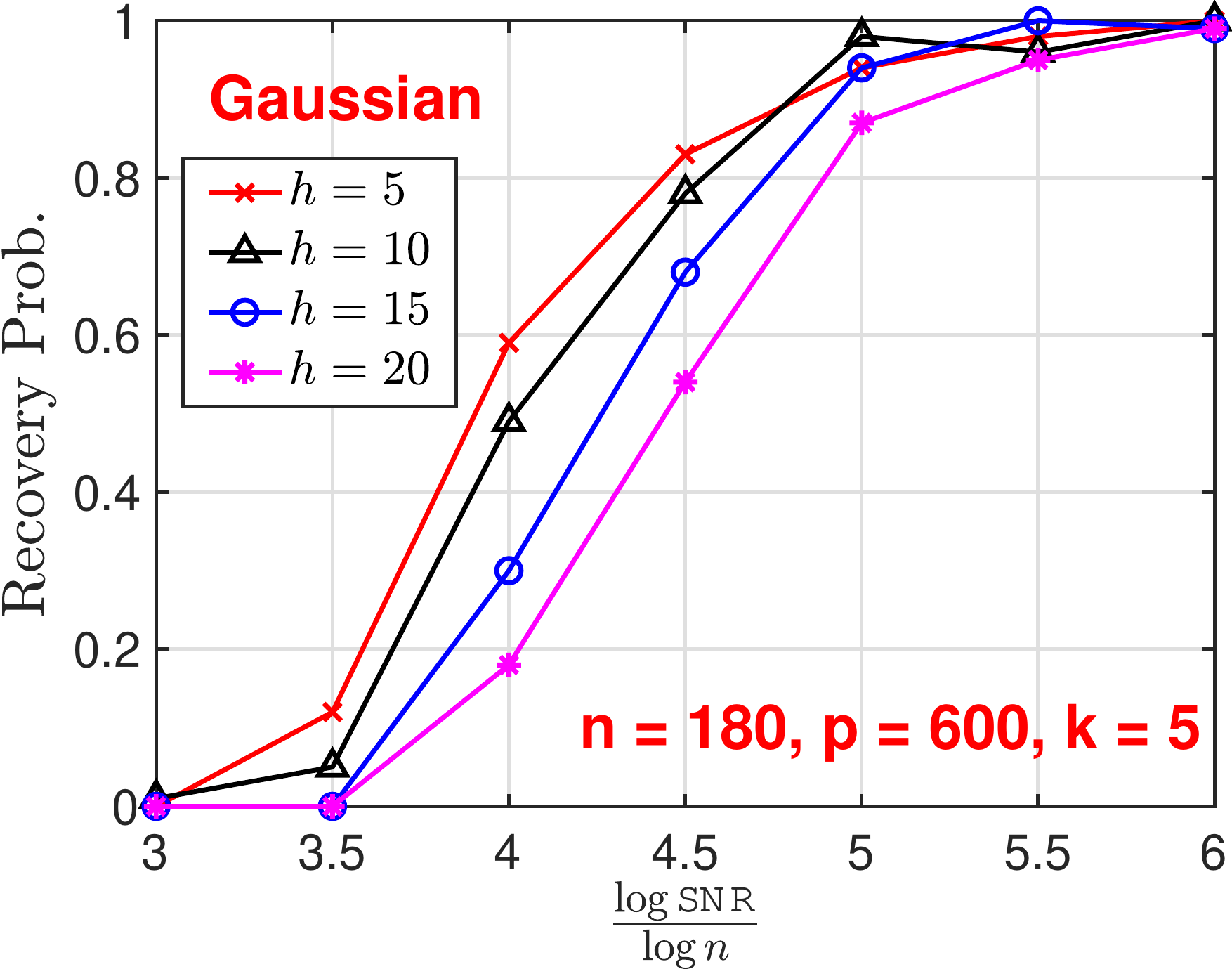}
\includegraphics[width = 2.75in]{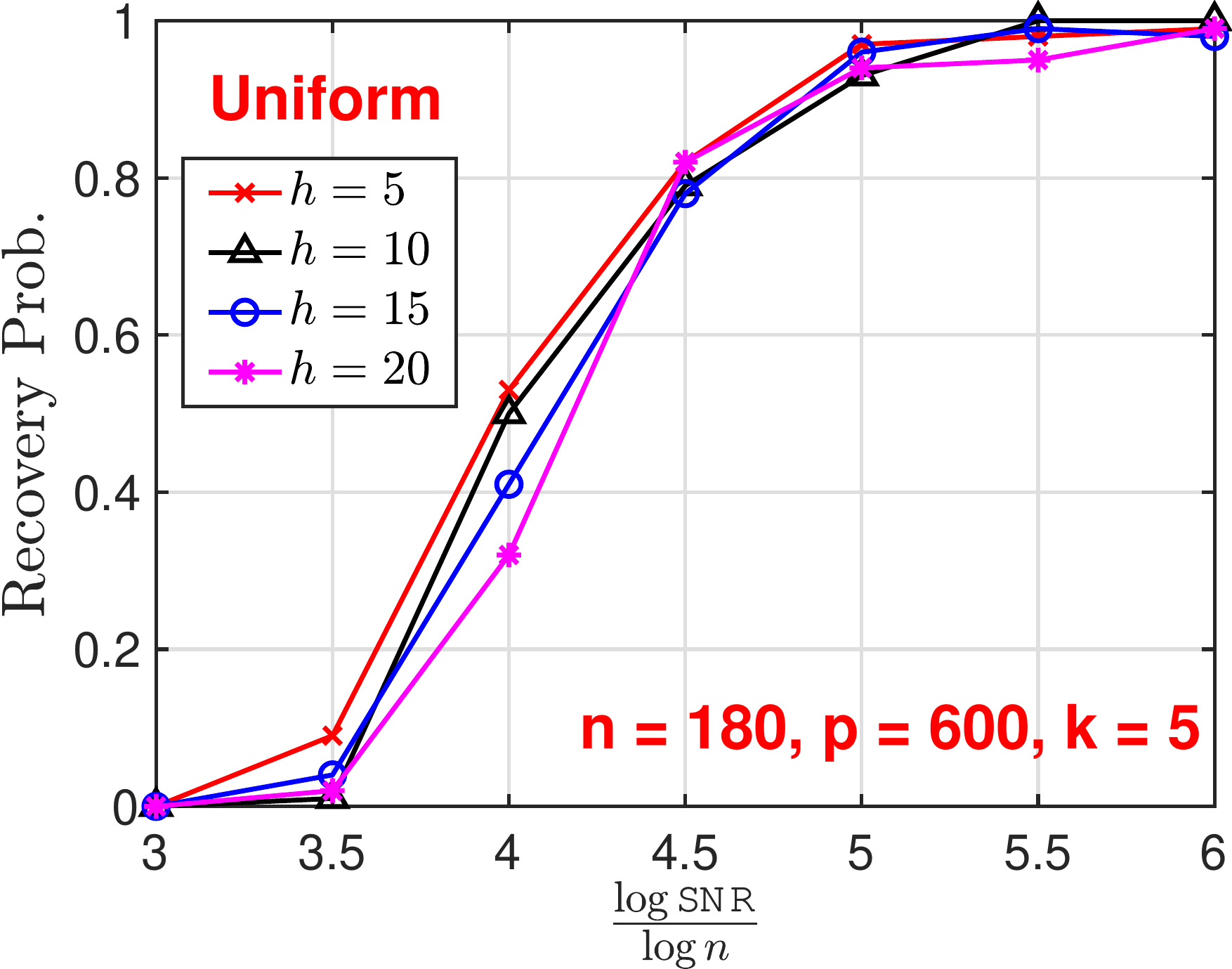}
}

\vspace{-0.05in}

\caption{Simulated permutation recovery rate $\Prob(\wh{\bPi} = \bPitrue)$ with
$n = \{120, 150, 180\}$, $p = 600$, $k = 5$, and $h = \{5, 10, 15, 20\}$, w.r.t. $\nfrac{\log \snr}{\log n}$.
(\textbf{Left Panel}) We have $\bX_{ij}$ be i.i.d. normal random variables, i.e., $\bX_{ij}\iid \normdist(0, 1)$;
(\textbf{Right Panel}) We have $\bX_{ij}$ be i.i.d. sub-gaussian random variables, to be more specific, $\bX_{ij} \iid \Unif[-1, 1]$.
}
\label{fig:h_impact}
\end{figure}

\subsection{Impact of permuted row number}

We investigate the impact of the permuted row number $h$ on the
simulated permutation recovery rate $\Prob(\wh{\bPi} = \bPitrue)$. We fix the signal length $p$ and
sparsity number $k$ as $600$ and $5$, respectively.
We let the sample number $n = \{120, 150, 180\}$ and
set the permuted row number $h\in \{5, 10, 15, 20\}$. The experiment results are presented
in Figure~\ref{fig:h_impact}.

\vspace{0.1in}\noindent
\textbf{Discussion.}
We notice that the permutation recovery becomes
more difficult, in other words, requires a larger
$\snr$, with an increasing number of permuted rows.
Under the Gaussian setting $(n, p, h, k) = (120, 600, 5, 5)$, we can obtain
the ground-truth $\bPitrue$
when $\log \snr \approx 5\log n$. When $h$ increases to
$20$, the requirement on $\snr$ is strengthened to
$\log \snr > 6\log n$.
We believe that this conclusion should hold universally.
However, numerical experiments do suggest that the performance difference
becomes less distinguishable with a higher $n/p$ ratio.

\section{Conclusion}\label{sec:conclusion}

We have studied sparse recovery
with shuffled labels.
First, we establish the statistical lower bounds
for both the sample number $n$ and $\snr$.
For the sample number $n$, by exploiting the sparsity of signals,
we manage to reduce the required sample number from
$n \geq 2p$ to the order of $\Omega(k\log p)$.
For $\snr$, we have a marginal increase
from  $\log\snr\gsim \log n$
to $\log\snr\gsim \log n + \nfrac{k}{n}\log(\frac{ep}{k})$.
Then, we present an exhaustive-search based estimator to confirm the
tightness of the above bounds.
Afterwards, we propose a practical estimator
and show they can yield the correct
 $(\bPitrue, \supp(\bbetatrue))$ under mild
  conditions.
Simulations confirm our theorems and suggest
that large sparsity number and Hamming distance
require more samples and stronger signal energy for
correct reconstruction
of $(\bPitrue, \supp(\bbetatrue))$.

\bibliographystyle{plainnat}
\bibliography{refs_scholar}

\begin{appendices}
\section{Proof of Statistical Lower Bound}

\subsection{Proof of Theorem~\ref{thm:statis_lb}}
\label{subsec:code_theory_explan}

\begin{proof}
To begin with, we assume $\bbetatrue \in \set{0, 1}^p$
and place a uniform distribution prior on
$\bPitrue$ and $\supp(\bbetatrue)$.
Then, we notice the relation
\begin{align}
\label{eq:minimax_error_prob_ub}
& \sup_{\bPitrue, \bbetatrue} \Expc_{\bX,\bw} \Ind\big[(\bPitrue, \supp(\bbetatrue)) \neq (\wh{\bPi}, \supp(\wh{\bbeta}))\big] \notag \\
\geq~& \Prob_{\bX,\bw, \bPitrue, \supp(\bbetatrue)}\big[(\bPitrue, \supp(\bbetatrue)) \neq (\wh{\bPi}, \supp(\wh{\bbeta}))\big] \defequal \vartheta,
\end{align}
where $\Expc_{\bX,\bw}(\cdot)$ denotes the expectation w.r.t.
$\bX$ and $\bw$, and $\Prob_{\bX,\bw, \bPitrue, \supp(\bbetatrue)}$ puts
uniform prior on $\bPitrue$ and $\supp(\bbetatrue)$ as well.
Since \eqref{eq:minimax_error_prob_ub} holds
universally, we can safely add
$\inf_{\wh{\bPi}, \wh{\bbeta}}$ to the left-hand side in~\eqref{eq:minimax_error_prob_ub} and
complete the proof. In the following context,
we lower bound the error probability
by adapting the techniques used in proving the
Fano's inequality in Theorem~$2.10.1$ in~\citet{cover2012elements}.

Denote $\entH(\cdot)$ as the entropy, while
$\entI(\cdot; \cdot)$ as the mutual information.
With the Fano's method as illustrated in~\citet{cover2012elements},
we would like to lower bound the error
probability $\Prob((\bPitrue, \supp(\bbetatrue)) \neq (\wh{\bPi}, \supp(\wh{\bbeta})))$ as
\begin{align}
& \entH(\bPitrue, \supp(\bbetatrue)) =
\entH\bracket{\bPitrue, \supp(\bbetatrue)~|~\bX}
\label{eq:statis_lb_entropy} \notag \\
=~& \entH\big(\bPitrue, \supp(\bbetatrue)~|~\bX, \wh{\bPi}, \supp(\wh{\bbeta})\big) + \entI(\bPitrue, \supp(\bbetatrue); \wh{\bPi}, \supp(\bbetatrue)~|~\bX) \notag \\
\stackrel{\cirone}{\leq}~&
\entH\big(\bPitrue, \supp(\bbetatrue)|\wh{\bPi}, \supp(\wh{\bbeta})\big) + \entI(\bPitrue, \supp(\bbetatrue); \wh{\bPi}, \supp(\wh{\bbeta})|\bX) \notag \\
\stackrel{\cirtwo}{\leq}~&
\entH\big(\bPitrue, \supp(\bbetatrue)~|~\wh{\bPi}, \supp(\wh{\bbeta})\big) + \
\entI\bracket{\bPitrue, \supp(\bbetatrue); \by~|~\bX} \notag \\
\stackrel{\cirthree}{\leq}~&1 + \log\big(|\pset_n|\times \binom{p}{k}\big)\vartheta
+ \entI\bracket{\bPitrue, \supp(\bbetatrue);\by|\bX},
\end{align}
where $\cirone$ is because of the property
such that conditioning reduces entropy~\cite[Eq.~(2.157)]{cover2012elements},
$\cirtwo$ is due to
the fact $(\bPitrue, \supp(\bbetatrue)) \rightarrow \
\by \rightarrow (\wh{\bPi}, \supp(\wh{\bbeta}))$ forms a Markov chain and
the data-processing inequality
\cite[Theorem~2.8.1]{cover2012elements}; and $\cirthree$ is a direct
consequence of Fano's inequality~\cite[Theorem~$2.10.1$]{cover2012elements}.
Exploiting the
independence between $\bPitrue$ and $\bbetatrue$,
we have
\[
\hspace{-0.02in}\entH(\bPitrue, \supp(\bbetatrue)) \hspace{-0.02in}=\hspace{-0.02in}\entH(\bPitrue) + \entH(\supp(\bbetatrue)) =
\log\abs{\pset_n} + \log{p\choose k}.
\]
Combing \eqref{eq:statis_lb_entropy}
with Lemma~\ref{lemma:statis_infor_lb} then
complete the proof.
\end{proof}

\newpage

\subsection{Proof of Theorem~\ref{thm:statis_approx_lb}}

\begin{proof}
We assume $\bPitrue$ is uniformly distributed over the
set $\pset_n$, which corresponds to the case where
no prior knowledge about $\bPitrue$ is unavailable.
First, we define $\calE \defequal \Ind\big\{\dH(\wh{\bPi}, \bPitrue) +\dH\bracket{\supp(\bbetatrue), \supp(\wh{\bbeta})}  \geq \mathsf{D}\big\}$,
which indicates the failure of approximate recovery of $\bPitrue$.
We give a roadmap before going into the details
\begin{itemize}
\item
\textbf{Step I:} We consider
the conditional  entropy
$\entH(\calE, \bPitrue, \supp(\bbetatrue)|~\wh{\bPi},  \supp(\wh{\bbeta}), \by, \bX)$ and prove
\[
& \entH(\calE, \bPitrue, \supp(\bbetatrue)~|~\wh{\bPi}, \supp(\wh{\bbeta}), \by, \bX)
= \entH(\bPitrue, \supp(\bbetatrue)~|~\by, \bX).
\]

\item
\textbf{Step II:}
We show that
\[
&\entH(\calE, \bPitrue, \supp(\bbetatrue)~|~\wh{\bPi}, \supp(\wh{\bbeta}), \by, \bX) \leq\log 2 + \entH(\bPitrue, \supp(\bbetatrue)) - \Prob(\calE = 0)\log \zeta,
\]
\par \noindent
where $\zeta$ is defined in \eqref{eq:approx_zeta_def}.

\item
\textbf{Step III:}
Combining the above two steps together, we upper-bound $\Prob(\calE = 0)$ as
\[
\Prob\bracket{\calE = 0} \leq~& \
\dfrac{\log 2 + \entI\bracket{\bPitrue, \supp(\bbetatrue); \by, \bX }}{\log\zeta}
\stackrel{\cirone}{=}
\dfrac{\log 2 + \entI\bracket{\bPitrue, \supp(\bbetatrue); \by~|~\bX }}{\log \zeta},
\]
where $\cirone$ is because $(\bPitrue, \supp(\bbetatrue))$
and $\by$ are independent given $\bX$.
Invoking Lemma~\ref{lemma:statis_infor_lb}, we
complete the proof.

\end{itemize}

Then we present the computational details.

\vspace{0.1in}\noindent
\textbf{Step I.}
We expand $H(\calE, \bPitrue, \supp(\bbetatrue)~|~\wh{\bPi}, \supp(\wh{\bbeta}), \by, \bX)$ via the
chain rule~\cite[Theorem~2.5.1]{cover2012elements} as
\[
&\entH(\calE, \bPitrue, \supp(\bbetatrue)~|~\wh{\bPi}, \supp(\wh{\bbeta}), \by, \bX) \\
=~& \entH(\bPitrue, \supp(\bbetatrue)~|~\wh{\bPi}, \supp(\wh{\bbeta}), \by, \bX) +
\entH(\calE~|~\bPitrue, \supp(\bbetatrue), \wh{\bPi},\supp(\wh{\bbeta}), \by, \bX)\\
=~&\entH(\bPitrue, \supp(\bbetatrue)~|~\by, \bX),
\]
where in the last equation we have used that $\cirone$
$\calE$ is deterministic
conditional on $\bPitrue, \supp(\bbetatrue), \wh{\bPi}, \by, \bX$,
and $\cirtwo$ $(\bPitrue, \bbetatrue)$ and $(\wh{\bPi}, \wh{\bbeta})$ are independent given $\bX$ and $\by$. \par

\vspace{0.1in}\noindent
\textbf{Step II.}
Define
$\BB\Bracket{\bracket{\bPitrue, \bbetatrue};\mathsf{D}}$
as
\[
\hspace{-0.05in}
\BB\Bracket{\bracket{\bPitrue, \bbetatrue}; \mathsf{D}} \hspace{-0.05in} \defequal\hspace{-0.07in} \
\set{
(\bPitrue, \supp(\bbetatrue))\Bigg|
\begin{aligned}
&~\dH(\wh{\bPi}, \bPitrue) = i, \\
&~\dH(\supp(\bbetatrue),\supp(\wh{\bbeta}))= j, \\
&\St~ i+ j \leq \mathsf{D}
\end{aligned}
}
\]
which denotes
the set of all possible pairs $(\bPitrue, \supp(\bbetatrue))$
given $\calE = 0$.
Easily we can verify that
its cardinality
can be upper bounded by
\[
\abs{\BB\Bracket{\bracket{\bPitrue, \bbetatrue};\mathsf{D}}}\leq \
\sum_{i=1}^{\mathsf{D}}\sum_{j=1}^{(\mathsf{D}-i)\vcap k}{n\choose i}i!\cdot {k \choose j}{p-k\choose j}.
\]
Then we expand
$H(\calE, \bPi^{*}|~\wh{\bPi}, \by, \bX)$
as
\[
& \entH\big(\calE, \bPitrue, \supp(\bbetatrue)~|~\wh{\bPi}, \supp(\wh{\bbeta}), \by, \bX\big) \\
=~& \entH(\calE|\wh{\bPi}, \supp(\wh{\bbeta}), \by, \bX) + \entH(\bPitrue, \supp(\bbetatrue)|\calE, \wh{\bPi}, \supp(\wh{\bbeta}), \by, \bX)\notag  \\
\stackrel{\cirtwo}{\leq} ~&\log 2 + \entH\big(\bPitrue, \supp(\bbetatrue)~|~\calE, \wh{\bPi}, \supp(\wh{\bbeta}), \by, \bX\big) \notag \\
\stackrel{\cirthree}{\leq}~& \log 2 + \
\Prob\left(\calE = 1\right) \entH(\bPitrue, \supp(\bbetatrue)|\calE = 1,\wh{\bPi},  \supp(\wh{\bbeta})) +\Prob\left(\calE = 0\right)\entH(\bPitrue, \supp(\bbetatrue)|\calE = 0,\wh{\bPi},  \supp(\wh{\bbeta})) \notag \\
\stackrel{\cirfour}{\leq}~& \log 2 + \left[1 - \Prob\left(\calE = 0\right)\right] \entH(\bPitrue, \supp(\bbetatrue)) + \Prob\left(\calE = 0\right)\log\Bracket{\sum_{i=1}^{\mathsf{D}} \sum_{j=1}^{(\mathsf{D}-i)\vcap k}\
\dfrac{n!}{(n-i)!}{k\choose j}{p-k \choose j} } \notag \\
\stackrel{\cirfive}{=} ~& \log 2 + \entH(\bPitrue, \supp(\bbetatrue)) -
\Prob\left(\calE = 0\right)\log \zeta,
\]
where in $\cirtwo$ we use the fact that $\calE$ is binary and hence
$\entH(\calE|\cdot) \leq \log 2$,
in $\cirthree$ we use the property that
conditioning reduces entropy~\citep[Equation~($2.157$)]{cover2012elements},
in $\cirfour$ we use the property
\[
\entH\big(\bPitrue, \supp(\bbetatrue)|\calE = 0,\wh{\bPi},  \supp(\wh{\bbeta})\big)\leq \log\abs{\BB\Bracket{\bracket{\bPitrue, \bbetatrue};\mathsf{D}}},
\]
and in $\cirfive$ we use the fact that $\entH(\bPitrue, \supp(\bbetatrue)) = \log\bracket{n!\cdot {p\choose k}}$.
\end{proof}

\subsection{Supporting lemmas for Section~\ref{sec:statis_lb}}
\begin{lemma}
\label{lemma:statis_infor_lb}
Assume that $\bbetatrue\in \set{0, 1}^p$.
we have
\[
\entI\bracket{\bPitrue, \supp(\bbetatrue); \by|\bX} \leq \
\dfrac{n}{2}\log\bracket{1 + \snr},
\]
where $\entI(\cdot; \cdot)$ denotes the mutual information.
\end{lemma}

\begin{proof}
Denote $\enth(\cdot)$ as the differential entropy. We have
\[
\entI\big(\bPitrue, \supp(\bbetatrue); \by~|~\bX\big)
\stackrel{\cirone}{=}~& \Expc_{\bX, \bw, \bPitrue, \supp(\bbetatrue)}\Bracket{\enth\bracket{\by|\bX = \bx} - \enth\bracket{\by|\bPitrue, \supp(\bbetatrue),\bX=\bx}} \\
\stackrel{\cirtwo}{\leq}~&
\Expc_{\bX}\dfrac{1}{2}\logdet\bracket{\Expc_{\bw, \bPitrue~|~\bX=\bx}\by \by^{\rmt}}\
 - \dfrac{n}{2}\log \sigma^2\\
\stackrel{\cirthree}{\leq}~& \dfrac{1}{2}\logdet\hspace{-0.01in}\bigg[\hspace{-0.01in}\Expc_{\bX, \bPitrue}(\sigma^2 \bI_{n\times n} + \bPitrue\bX\bbetatrue\
\bbeta^{\natural\rmt}\bX^{\rmt}\bPi^{\natural\rmt} )\bigg] - \dfrac{n}{2}\log \sigma^2 \\
\stackrel{\cirfour}{=}~&\dfrac{n}{2}\log\bracket{\sigma^2 + \|\bbetatrue\|_{2}^2} - \dfrac{n}{2}\log\sigma^2 = \frac{n}{2}\log(1+ \snr),
\]
where $\cirone$ is because of the definition of conditional
mutual information;
$\cirtwo$ is due to the property~\cite[Theorem~8.6.5]{cover2012elements}
\[
\enth(\bZ) \leq \dfrac{1}{2}\logdet\Cov(\bZ) \leq \
\dfrac{1}{2}\logdet \Expc\bracket{\bZ\bZ^{\rmt}},
\]
for a random variable $\bZ$ with finite covariance matrix $\Cov(\bZ)$,
and $\enth(\by|\bPitrue, \supp(\bbetatrue), \bX=\bx) = \enth(\bw)$ as $\bbetatrue$'s
information is fully encoded in $\supp(\bbetatrue)$;
in $\cirthree$ we use the
concavity of $\logdet(\cdot)$, i.e., $\Expc\logdet(\cdot) \leq \logdet \Expc(\cdot)$;
and in $\cirfour$ we have
\[
\Expc_{\bX, \bPitrue}\bracket{\bPitrue\bX\bbetatrue\
\bbeta^{\natural\rmt}\bX^{\rmt}\bPi^{\natural\rmt}} = \
\|\bbetatrue\|^2_2 \cdot \bI_{n\times n}.
\]
\end{proof}

\section{Analysis of ML estimator}
This section analyzes the ML estimator.

\subsection{Notation definition}
\label{subsec:ml_events}
We denote $\supp(\bbetatrue)$ and $\supp(\bbeta)$
as $T$ and $S$, respectively.
In addition, we define $\calT_1$, $\calT_2$, and $\calT_3$ as
$\calT_1 \defequal T\setminus S$,
$\calT_2 \defequal T\bigcap S$, and
$\calT_3 \defequal S\setminus T$, respectively. An illustration is available
in Figure~\ref{fig:supp_span}.
\begin{figure}[!ht]
\centering
\includegraphics[width = 2.75in]{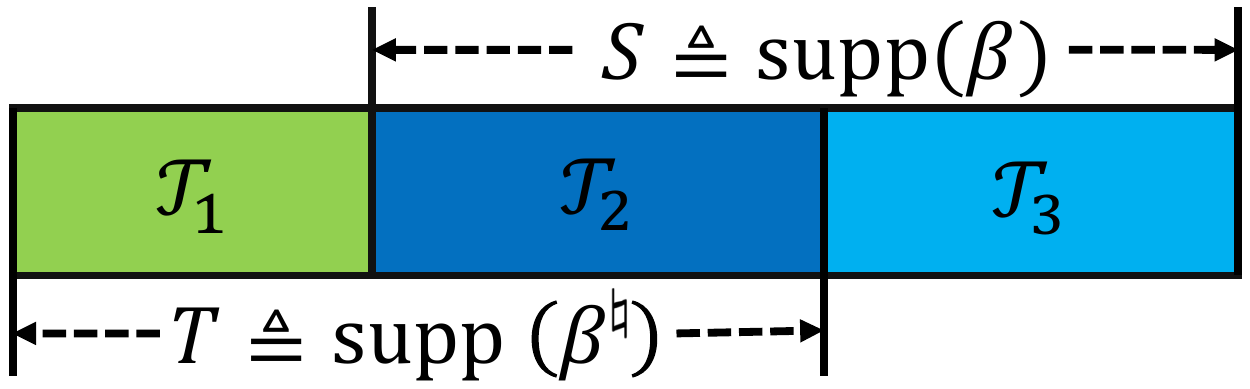}
\caption{Illustration of $\calT_1$, $\calT_2$, and $\calT_3$:
$\calT_1 \defequal T\setminus S$, $\calT_2 \defequal T\bigcap S$,
and $\calT_3 \defequal S\setminus T$, respectively.}
\label{fig:supp_span}
\end{figure}
In addition, we define the following events before we
proceed.
\[
\calE_1 &\defequal \set{\bPitrue = \wh{\bPi}, \
\bbetatrue \neq \wh{\bbeta}}; \notag \\
\calE_2 &\defequal \set{\bPitrue \neq \wh{\bPi}, \
\bbetatrue = \wh{\bbeta}}; \notag \\
\calE_3 &\defequal \set{\bPitrue \neq \wh{\bPi}, \
\bbetatrue \neq \wh{\bbeta}}; \notag \\
\calE_4(\delta; \bPi, S) &\defequal \
\set{\norm{\Proj_{\bPi\bX_S}^{\perp}\by}{2}^2 - \
\norm{\Proj_{\bPi\bX_S}^{\perp}\bw}{2}^2< 2\delta} ; \notag  \\
\calE_5(\delta; \bPi, S) &\defequal \
\set{\abs{\norm{\Proj_{\bPitrue\bX_T}^{\perp}\bw}{2}^2  - \norm{\Proj_{\bPi\bX_S}^{\perp}\bw}{2}^2}\geq \delta}; \notag  \\
\calE_6(t; h)&\defequal \bigg\{\norm{\Proj_{\bPi\bX_S}^{\perp}\bPitrue\bX_{T}\bbeta^{\natu}_T}{2}^2\geq t \|\bbetatrue_T\|_2^2,~\forall~S, \bPi~\textup{s.t.}~\dH(\bI, \bPi) = h\bigg \}; \notag  \\
\calE_7(\delta; \bPi, S) & \defequal \set{\abs{\norm{\Proj_{\bPitrue\bX_T}^{\perp}\bw}{2}^2 - \Expc \norm{\Proj_{\bPitrue\bX_T}^{\perp}\bw}{2}^2}\
\geq \frac{\delta}{2} }; \notag \\
\calE_8(\delta; \bPi, S) & \defequal \set{\abs{\norm{\Proj_{\bPi\bX_S}^{\perp}\bw}{2}^2 - \Expc \norm{\Proj_{\bPi\bX_S}^{\perp}\bw}{2}^2}\geq \frac{\delta}{2}}.
\]

\subsection{Proof of Theorem~\ref{thm:noiseless}}
\label{thm_proof:noiseless}
We prove a more specific version of Theorem~\ref{thm:noiseless}, which is

\begin{theorem*}
Consider the noiseless case, i.e., $\bw = \bZero$.
Suppose
$n = \Omega(k\log p)$, we have
$\Prob((\wh{\bPi}_{\ml}, \wh{\bbeta}_{\ml})\neq (\bPitrue,~\bbetatrue))
\lsim n^{-2}$.
\end{theorem*}

\begin{proof}
We upper bound
the error probability $\Prob(\wh{\bPi}_{\ml}\neq \bPitrue)$ can be
decomposed as
$\sum_{\ell=1}^3 \Expc\Ind(\calE_{\ell})$, whose
definitions are stated in Subsection~\ref{subsec:ml_events}.

\vspace{0.1in}
\noindent
\textbf{Bounding $\Expc\Ind(\calE_1)$.}
Conditional on $\calE_1$, easily we can prove that $\calT_1$ is not empty, since otherwise
we have $\bX_T \bbetatrue_T = \bX_T\wh{\bbeta}$, which leads to contradiction.
Given the support set $S$, we can write
$\bbeta_S$ as
$(\bX_S^{\rmt}\bX_S)^{-1}\bX_S^{\rmt}\bX_T\bbetatrue_T$.
Then we have
\[
\bX_T \bbetatrue_T =
\bX_S \bbeta_S = \
\Proj_{\bX_S} \bX_T\bbetatrue_T,
\]
which implies $\bX_T\bbetatrue_T$ lies
within the linear space spanned by the columns of
$\bX_S$. Then we obtain
$\Proj_{\bX_S}^{\perp}(\bX_T\bbetatrue_T) = \
\Proj_{\bX_S}^{\perp}(\bX_{\calT_1}\bbetatrue_{\calT_1}) = 0$ and
can upper-bound $\Expc\Ind(\calE_1)$ as
${p\choose k}\cdot \Prob(\|\Proj_{\bX_S}^{\perp}(\bX_{\calT_1}\bbetatrue_{\calT_1})\|_{2} = 0)$.
Recalling the definition of $\calT_1$, we conclude
$\bX_{\calT_1}$ to be independent of $\bX_S$.
Then, with the rotational invariance of Gaussian distribution,
we can rewrite $\Prob(\|\Proj_{\bX_S}^{\perp}(\bX_{\calT_1}\bbetatrue_{\calT_1})\|_{2} = 0)$
as $\Prob(\norm{\Proj_{\bX_S}^{\perp}\bz}{2}
\|\bbetatrue_{\calT_1}\|_2=0) = \Prob(\norm{\Proj_{\bX_S}^{\perp}\bz }{2} =0)$, where $\bz \in \RR^n$ is
a Gaussian random variable satisfying $\bz \sim \normdist(0, \bI_{n\times n})$.

In the following proof, we will view $\bz$ as
a fixed vector as it is independent of $\bX_{S}$.
Afterwards, we can upper-bound $\Expc\Ind(\calE_1)$
as
\[
\Expc \Ind(\calE_1) \leq~& {p\choose k}
\Prob\bracket{\norm{\Proj_{\bX_S}^{\perp}\bz }{2}\leq \gamma \|\bz\|_2}
\stackrel{\cirone}{\leq}
\bracket{\frac{ep}{k}}^k \exp\Bracket{\frac{n-k}{n}\bracket{1-\gamma^2 + \log \gamma^2}}
\stackrel{\cirtwo}{\leq} n^{-2},
\]
where in $\cirone$ we use ${p\choose k}\leq (\nfrac{ep}{k})^k$ and
\cite[Lemma~2.2(a)]{dasgupta2003elementary},
and in $\cirtwo$ we pick $\gamma$ as
$\frac{k^{k}}{\sqrt{e}n^{2}(ep)^{k}}$.

\vspace{0.1in}
\noindent
\textbf{Bounding $\Expc\Ind\bracket{\calE_2}$.}
Event $\calE_2$ suggests that there exists
another permutation matrix $\wh{\bPi} \neq \bPitrue$ such that
$\wh{\bPi}\bX_T\bbeta_T^{\natural} = \bPitrue\bX_T\bbeta_T^{\natural}$.
Then, we can equate $\Expc\Ind(\calE_2)$
with $\Prob(\|(\wh{\bPi} - \bPitrue)\bz\|_2 = 0,~\exists~\wh{\bPi}\neq \bPitrue)$, where
$\bz \sim \normdist(\bZero, \bI_{n\times n})$.
This leads to
\[
& \Expc\Ind(\calE_2) \leq \Prob(\|(\wh{\bPi} - \bPitrue)\bz\|_2 = 0,~\exists~\wh{\bPi}\neq \bPitrue) \leq \sum_{h\geq 2}
{n\choose h}h!\cdot \Prob(\|(\bI-\bPi)\bz\|_2 = 0,\St\dh(\bI, \bPi) = h).
\]
With Lemma~\ref{lemma:permute_cancel_tail_bound}, we have
\[
\Expc\Ind(\calE_2) \leq &
 \sum_{h\geq 2}{n\choose h}h! \Prob(\|(\bI-\bPi)\bz\|_2 \leq \frac{4}{en^{20}},\St\dh(\bI, \bPi) = h)\notag \\
\stackrel{\cirthree}{\leq} &~ 6\sum_{h\geq 2} n^h \cdot \exp\Bracket{\frac{h}{10}\big (\log\big (\nfrac{2}{(ehn^{20})} \big ) -\nfrac{2}{(ehn^{20})} + 1\big)}
\stackrel{\cirfour}{\leq}  6\sum_{h\geq 2} n^{-h} \leq \frac{6}{n(n-1)},
\]
where $\cirthree$ is because $n!/(n-h)! \leq n^h$,
and $\cirfour$ is due to $\exp\Bracket{\frac{h}{10}\big (\log\big (\nfrac{2}{(ehn^{20})} \big ) -\nfrac{2}{(ehn^{20})} + 1\big)}\leq n^{-2h}$.

\vspace{0.1in}\noindent
\textbf{Bounding $\Expc\Ind(\calE_3)$.}
Adopting the similar argument as in bounding $\Expc\Ind(\calE_1)$,
for a fixed permutation matrix $\bPi$ and support
set $S$, we have
\[
\wh{\bbeta} = \wh{\bbeta}_S = \
\bracket{\bPi \bX_{S}}^{\dagger} \
\bPitrue\bX_T\bbetatrue_T.
\]
Based on the optimality of the objective function
in \eqref{eq:ml_est}, we conclude that
\[
\bPitrue\bX_T\bbetatrue_T = \
\bPi\bX_S\wh{\bbeta} = \
\Proj_{\bPi\bX_S}\bPitrue\bX_T\bbetatrue_T,
\]
which suggests $\bPitrue\bX_T\bbetatrue_T$
lies within the linear space spanned by the columns in $\bPi \bX_S$.
Following a similar procedure as in bounding
$\Expc\Ind(\calE_1)$, we have
$\Expc \Ind(\calE_3) \leq \frac{8}{n(n-1)}$.
Combining the discussions thereof then completes the proof.

\end{proof}

\newpage

\subsection{Proof of Theorem~\ref{thm:noisy_ml}}
\label{thm_proof:noisy_ml}

\begin{proof}
The proof consists of
two stages
\begin{itemize}
\item
\textbf{Stage I.} The permutation matrix
can be obtained with high probability;

\item
\textbf{Stage II.} Given the correct permutation, we can
detect the support set with high probability.
\end{itemize}
The proof of Stage I is  in Lemma~\ref{lemma:ml_correct_permute}
while the proof of Stage II is  in Lemma~\ref{lemma:ml_sign_consistency}.

\end{proof}

\subsection{Supporting lemmas for Section~\ref{sec:ml_estim}}

\begin{lemma}
\label{lemma:ml_correct_permute}
Assume $n = \Omega\bracket{k\log\bracket{\frac{ep}{k}}}$,
and
\begin{align}
\label{eq:noisy_ml_assump}
\log(\snr) \geq~& \Bracket{128\log n + \nfrac{64k}{n}\cdot\log\bracket{\nfrac{ep}{k}}}
\vcup \Bracket{34\log n + 2\log(4e^6)} \vcup \Bracket{148\log n + 4\log(2e^6)},
\end{align}

\noindent
we conclude that
$\Prob(\wh{\bPi}_{\textup{ML}} \neq \bPitrue) \leq \frac{13}{n(n-1)}$.
\end{lemma}

\begin{proof}
Fixing the support set $S$ and permutation matrix $\bPi$,
we have
$\min_{\supp(\wh{\bbeta}) = S}\|\by - \bPi \bX\bbeta\|_2$
be expressed as $\norm{\Proj_{\bPi\bX_S}^{\perp}\by}{2}$.
Here we only study the permutation
reconstruction error (i.e.,
$\wh{\bPi} \neq \bPitrue$) and define error event $\calE$ as
\[
\calE \defequal
\set{
\exists~\bracket{\bPi, S}~\textup{s.t.}~\bPi\neq \bPitrue,~\norm{\Proj_{\bPi\bX_S}^{\perp}\by}{2} \leq \
\norm{\Proj_{\bPitrue\bX_T}^{\perp}\by}{2}}.
\]
Then we would like to show $\calE$ holds with probability
close to zero.

First, we would like to show
$\calE \subseteq \cup_{\bPi, S}\calE_4(\delta; \bPi, S)\cup \calE_5(\delta; \bPi, S)$.
This is because
\[
& \bigcap_{\bPi, S}\Bracket{\br{\calE}_4(\delta; \bPi, S)\bigcap \br{\calE}_5(\delta; \bPi, S)} \\
=& \bigg\{\forall~(\bPi, S)~\textup{s.t.}~\bPi\neq \bPitrue, \
\norm{\Proj_{\bPi\bX_S}^{\perp}\by}{2}^2 - \
\norm{\Proj_{\bPi\bX_S}^{\perp}\bw}{2}^2\geq 2\delta, \abs{\norm{\Proj_{\bPitrue\bX_T}^{\perp}\bw}{2}^2  - \norm{\Proj_{\bPi\bX_S}^{\perp}\bw}{2}^2}< \delta\bigg \} \\
\subseteq & \hspace{-0.01in}\set{\forall~(\bPi, S)~\textup{s.t.}\bPi\neq \bPitrue, \
\norm{\Proj_{\bPi\bX_S}^{\perp}\by}{2}^2 \hspace{-0.01in}- \hspace{-0.01in}
\norm{\Proj_{\bPitrue\bX_T}^{\perp}\bw }{2}^2 \geq \delta} \subseteq\br{\calE}.
\]
Then we can bound $\Expc\Ind(\calE)$ as
\begin{align}
\Expc\Ind(\calE) \leq \
\sum_{h  \geq 2}{n\choose h}h!
\Bracket{\sum_{S}(\zeta_1 +
\zeta_2) + \zeta_3},
\label{eq:noisy_psi}
\end{align}
where $\zeta_1$, $\zeta_2$, and $\zeta_3$ are
defined as
\[
\zeta_1 \defequal~&~\Expc\Ind\big(\calE_4\bracket{\delta; \bPi, S} \bigcap \calE_6(t_h; h)~\big|~ \dH(\bI, \bPi) = h\big); \\
\zeta_2 \defequal~&~ \Expc\Ind\big(\calE_5\bracket{\delta; \bPi, S} \bigcap \calE_6(t_h; h)~\big|~ \dH(\bI, \bPi) = h\big); \\
\zeta_3 \defequal~&~
\Expc\Ind\bracket{\br{\calE}_6(t_h; h)}.
\]

\newpage

The following analysis can be roughly divided into
$2$ steps
\begin{itemize}
\item \textbf{Step I.}
Setting $\delta = \nfrac{\big\|\Proj_{\bPi \bX_S}^{\perp}\bPitrue\bX_T \bbetatrue_T\big\|_{2}^2}{4}$,
we separately bound $\zeta_1$, $\zeta_2$, and
$\zeta_3$.

\item \textbf{Step II.}
We pick $t_h$ as $\nfrac{(n\cdot \log \snr)}{\snr}$ and show
$\Prob(\wh{\bPi}\neq \bPitrue) \leq \nfrac{13}{n(n-1)}$
provided assumption \eqref{eq:noisy_ml_assump} holds.
\end{itemize}

\noindent \textbf{Step I.}
Picking $\delta$ as $\nfrac{\big\|\Proj_{\bPi \bX_S}^{\perp}\bPitrue\bX_T \bbetatrue_T\big\|_{2}^2}{4}$, we have
\[
\zeta_1 \hspace{-0.03in}=~& \hspace{-0.03in} \Expc_{\bX, \bw}
\Ind\hspace{-0.02in}\Bracket{
\big \langle\Proj_{\bPi \bX_S}^{\perp}\bPitrue\bX_T \bbetatrue_T, \bw \big \rangle \leq
-\nfrac{\norm{\Proj_{\bPi \bX_S}^{\perp}\bPitrue\bX_T \bbetatrue_T}{2}^2}{4} \cap \calE_6(t_h; h)} \\
\stackrel{\cirone}{\leq}~& \Expc_{\bX} \Ind
\Bracket{\Phi\bracket{-\nfrac{\norm{\Proj_{\bPi \bX_S}^{\perp}\bPitrue\bX_T \bbetatrue_T}{2}^2}{4\sigma \norm{\Proj_{\bPi \bX_S}^{\perp}\bPitrue\bX_T \bbetatrue_T}{2}}} \bigcap \calE_6(t_h; h)} \\
\stackrel{\cirtwo}{\leq}~& \Expc\Phi\bracket{-\nfrac{\sqrt{t_h}\|\bbetatrue_T\|_{2}}{4\sigma}}
\stackrel{\cirthree}{\leq} \exp\bracket{-\nfrac{t_h \|\bbetatrue_T\|^2_2}{32\sigma^2}} = \exp\bracket{-\nfrac{t_h\cdot \snr}{32}},
\]
where in $\cirone$ we condition on $\bX$ and
have $2\langle\Proj_{\bPi \bX_S}^{\perp}\bPitrue\bX_T \bbetatrue_T, \bw \rangle$ to be a
Gaussian random variable with zero mean and $4\sigma^2\|\Proj_{\bPi \bX_S}^{\perp}\bPitrue\bX_T \bbetatrue_T\|_{2}^2$ variance, i.e.,
$\normdist(0, 4\sigma^2\|\Proj_{\bPi \bX_S}^{\perp}\bPitrue\bX_T \bbetatrue_T\|_{2}^2)$,
$\Phi(\cdot)$ denotes
the CDF for the standard normal distribution, $\cirtwo$
is because of event $\calE_6(t_h; h)$,
and $\cirthree$ is due to the tail bound
$\Phi(-x) \leq e^{-x^2/2},~x\geq 0$ (c.f. Proposition~$2.12$ in~\citet{vershynin2018high}).

For $\zeta_2$, we notice the
relation such that $\Expc_{\bw, \bX}\|\Proj_{\bPitrue\bX_T}^{\perp}\bw\|_{2}^2 = \
\Expc_{\bw, \bX} \|\Proj_{\bPi\bX_S}^{\perp}\bw\|_{2}^2 = (n-k)\sigma^2$
and perform the decomposition
\[
\calE_5(\delta; \bPi, S) \subseteq
\calE_7(\delta; \bPi, S)\bigcup \calE_8(\delta; \bPi, S).
\]
This leads to
\[
\zeta_2 \leq
\Expc\Ind(\calE_7(\delta; \bPi, S) \bigcap \calE_6(t_h; h))
+\Expc\Ind(\calE_8(\delta; \bPi, S) \bigcap \calE_6(t_h; h)).
\]
Exploiting the independence between $\bX$ and $\bw$,
we can condition on $\bX$ view
$\nfrac{\|\Proj_{\bPitrue\bX_T}^{\perp}\bw\|_{2}^2}{\sigma^2}$
($\nfrac{\|\Proj_{\bPi\bX_S}^{\perp}\bw\|_{2}^2}{\sigma^2}$ resp.) as
$\chi^2$ random variable with freedom $n-k$.
Plugging $\delta$ into the tail bound for
$\chi^2$  as in~\citet[Example 2.11, P.~29]{wainwright_2019}, we conclude
\[
\zeta_2 \leq 4\exp\Bracket{-\bracket{\frac{t_h\cdot \snr}{64} \vcap \frac{t_h^2 \cdot \snr^2}{512(n-k)}}}.
\]
Ultimately, we invoke Lemma~\ref{lemma:depend_proj_small_ball}
and bound $\zeta_3$ as
\[
\zeta_3  \leq ~&\
2n^{-2h} + 6 \exp\Bracket{\frac{h}{10}\bracket{\log\bracket{\frac{4e^6 n^{16h/n}t_h}{h}} -\frac{4e^6 n^{16h/n}t_h}{h}+ 1}},
\]
where $t_h < h/(4e^6 n^{16h/n})$.

\vspace{0.1in}\noindent
\textbf{Step II.}
Let $t_h$ be $\nfrac{(n\cdot \log \snr)}{\snr}$, we
will show $\zeta_{\ell}$ $(1\leq \ell \leq 3)$ all shrink to zero with the
assumption (\ref{eq:noisy_ml_assump}).
For $\zeta_1$, we
 use the assumption
$\log(\snr) \geq 128\log n + \nfrac{64k\log\bracket{\frac{ep}{k}}}{n} \geq \nfrac{128h\log n}{n} + \nfrac{64k\log(\frac{ep}{k})}{n}$
in \eqref{eq:noisy_ml_assump}
and have $\zeta_1 \leq
\exp\bracket{-\nfrac{t_h \cdot \snr}{32}}
\leq n^{-2h} \bracket{\nfrac{k}{ep}}^k$.

Then we turn to $\zeta_2$. If
$\frac{t_h\cdot \snr}{64} \vcap \frac{t_h^2 \cdot \snr^2}{512(n-k)} = \frac{t_h\cdot \snr}{64}$,
we can apply the same strategy to bound $\zeta_2$.
Otherwise, we need to bound $\zeta_2$ as
\begin{align}
\label{eq:noisy_psi5}
& \exp\bracket{-\frac{t_h^2 \cdot \snr^2}{512(n-k)}} = \
\exp\bracket{-\frac{n^2\cdot  \log^2 \snr}{512(n-k)}}
\stackrel{\cirfive}{\leq} \exp\bracket{-\frac{n\cdot \log\snr}{64}} \leq  \
n^{-2h} \bracket{\frac{k}{ep}}^k,
\end{align}
where in $\cirfive$ we use the relation $\log\snr\geq 8$ and redo the
above analysis. For $\zeta_3$, we first need to check the
condition $t_h <  h/(4e^6 n^{16h/n})$ is not violated as
\[
\log t_h < -16\log n - \log(4e^6) <
\log h - \nfrac{16h\log n}{n}- \log(4e^6).
\]
Then we consider $\exp[\frac{h}{10}(\log\nfrac{(4e^6 n^{16h/n}t_h)}{h} -\nfrac{(4e^6 n^{16h/n}t_h)}{h}+ 1)]$, which
reads as
\begin{align}
\label{eq:noisy_psi6}
 &\exp\Bracket{\frac{h}{10}\bracket{\log\bracket{\frac{4e^6 n^{16h/n}t_h}{h}} -\frac{4e^6 n^{16h/n}t_h}{h}+ 1}}\notag  \\
 \leq~& \
 \exp\Bracket{\frac{h}{10} \log\bracket{\frac{4e^6 n^{16h/n+1}}{h}}}
\times  \exp\Bracket{\frac{h}{10}\bracket{\log\frac{\log\snr}{\snr} - \frac{\log\snr}{\snr} + 1} }  \notag \\
\stackrel{\cirsix}{\leq}~& \exp\Bracket{\frac{h}{10} \log\bracket{\frac{4e^6 n^{16h/n+1}}{h}}}
\times  \exp\bracket{-\frac{h\cdot \log\snr}{40}}
\stackrel{\cirseven}{\leq} n^{-2h},
\end{align}
where in $\cirsix$ we use $\frac{\log z}{z} - 1 - \log\frac{\log z}{z}\geq \frac{\log z}{4}$ when $z\geq 1.25$,
and in $\cirseven$ we use the relation
$\log\snr \geq 148\log n + 4\log(2e^6)$ in
\eqref{eq:noisy_ml_assump}.

Combining \eqref{eq:noisy_psi},
\eqref{eq:noisy_psi5}, and
\eqref{eq:noisy_psi6}, we complete the proof as
\[
\Expc\Ind(\calE) \leq~& \sum_{h\geq 2} {n\choose h}h! \cdot
\Bracket{\sum_{S} \bracket{5n^{-2h} \bracket{\nfrac{k}{ep}}^k} + 8n^{-2h}} \leq 13\sum_{h\geq 2}n^{-h} \leq \frac{13}{n(n-1)}.
\]
\end{proof}

\begin{lemma}
\label{lemma:ml_sign_consistency}
Consider the case where the ground-truth permutation matrix
$\bPitrue$ is given a prior. Provided that $(i)$ $n\gsim k\log p$,
and $(ii)$ $\min_{i\in T} \nfrac{|\bbeta^{\natural}_i|^2}{\sigma^2}
\gsim 1$, we have
$\Prob(S \neq T) \leq c_0 e^{-c_1 k\log p}$,
where $c_0, c_1 > 0$ are some positive constants.	
\end{lemma}

\begin{proof}
We assume $\bPi^{\natural} = \bI$ w.l.o.g.
Recalling the definition of our estimator in \eqref{eq:ml_est},
we ought to have
\[
\fnorm{\Proj_{\bX_S}^{\perp} \by}^2 \leq \
\fnorm{\Proj_{\bX_T}^{\perp} \by}^2 =
\fnorm{\Proj_{\bX_T}^{\perp} \bw}^2,
\]
which is equivalent to
\[
\fnorm{\Proj_{\bX_S}^{\perp}\bX_T\bbeta^{\natural}_T }^2
+
2 \la \Proj_{\bX_S}^{\perp} \bX_T\bbeta^{\natural}_T, \bw\ra
\leq \
\fnorm{\Proj_{\bX_T}^{\perp} \bw}^2 -
\fnorm{\Proj_{\bX_S}^{\perp}\bw}^2.
\]
To begin with, we consider a fixed index set
$S$ and have
\begin{align}
\label{eq:ml_sign_consistency_tot}
\Prob\bracket{\fnorm{\Proj_{\bX_S}^{\perp} \by}
\leq \fnorm{\Proj_{\bX_T}^{\perp} \by}}
\leq~ \zeta_1 + \zeta_2 + \zeta_3,
\end{align}
where $\zeta_1$, $\zeta_2$, and $\zeta_3$
are defined as
\[
\zeta_1 \defequal~& \Prob\big(\|\Proj_{\bX_S}^{\perp} \bX_{\calT_1} \bbeta^{\natural}_{\calT_1}\|_2^2
\leq \nfrac{(n-k)}{4}
\|\bbeta^{\natural}_{\calT_1}\|^2_2\big); \\
\zeta_2 \defequal~& \Prob\big(2\langle\Proj_{\bX_S}^{\perp} \bX_T\bbeta^{\natural}_T, \bw \rangle
\lsim -c\sigma \sqrt{(k\log p)(n-k)}\|\bbeta^{\natural}_{\calT_1}\|_2\big); \\
\zeta_3 \defequal~& \Prob\big(\fnorm{\Proj_{\bX_T}^{\perp} \bw}^2 -
\fnorm{\Proj_{\bX_S}^{\perp}\bw}^2
\geq \nfrac{(n-k)}{4}
\|\bbeta^{\natural}_{\calT_1}\|^2_2 - c\sigma \sqrt{(k\log p)(n-k)}\|\bbeta^{\natural}_{\calT_1}\|_2\big).
\]
\par \noindent
\textbf{Analysis of $\zeta_1$.}
Recalling the fact such that $\Proj_{\bX_S}^{\perp}$
is the projection onto the orthogonal complement
of column space spanned by $\bX_T$, we can verify
$\Proj_{\bX_S}^{\perp}\bX_T\bbeta^{\natural}_T
= \Proj_{\bX_S}^{\perp} \bX_{\calT_1} \bbeta^{\natural}_{\calT_1}$.
Then, we decompose $\zeta_1$ as
\begin{align}
\label{eq:ml_sign_consistency_zeta1}
\zeta_1\leq \zeta_{1,1} + \zeta_{1, 2},
\end{align}
where $\zeta_{1,1}$ and $\zeta_{1,2}$ are defined as
\[
\zeta_{1,1} \defequal ~&\Prob\bigg(\|\Proj_{\bX_S}^{\perp} \bX_{\calT_1} \bbeta^{\natural}_{\calT_1}\|_2
\leq  \nfrac{\sqrt{n-k}}{2}
\|\bbeta^{\natural}_{\calT_1}\|_2,  \|\bX_{\calT_1} \bbeta^{\natural}_{\calT_1} \|_2 \geq \sqrt{\nfrac{n}{2}}\|\bbeta^{\natural}_{\calT_1}\|_2 \bigg); \\
\zeta_{1,2}\defequal ~&
\Prob\bracket{\|\bX_{\calT_1} \bbeta^{\natural}_{\calT_1} \|_2 \leq \sqrt{\nfrac{n}{2}}\|\bbeta^{\natural}_{\calT_1}\|_2}.
\]
For $\zeta_{1,1}$, we follow the same procedure as
in Theorem~\ref{thm:noiseless}. Conditional on $\bX_{\calT_1}$,
we view $\Proj_{\bX_S}^{\perp}$
as a random projection from a linear space of
dimension $n$ to a linear space of dimension
$n-k$, which yields
\begin{align}
\label{eq:ml_sign_consistency_zeta11}
\zeta_{1,1} \leq~&
\Prob\bracket{\|\Proj_{\bX_S}^{\perp} \bX_{\calT_1} \bbeta^{\natural}_{\calT_1}\|_2^2
\leq \nfrac{(n-k)}{2n} \|\bX_{\calT_1}\bbeta^{\natural}_{\calT_1}\|^2_2 }
\stackrel{\cirone}{\leq} \exp\Bracket{\frac{n-k}{2}(\log \nfrac{1}{2} - \nfrac{1}{2} + 1)} = e^{-c_0n},
\end{align}
where $\cirone$ is due to Lemma~\ref{lemma:random_proj}.
For $\zeta_{1,2}$, we exploit the fact
such that $\nfrac{\|\bX_{\calT_1} \bbeta^{\natural}_{\calT_1} \|_2^2}{\|\bbeta^{\natural}_{\calT_1}\|^2_2}$  is a $\chi^2$ random variable with freedom $n$.
Then we have
\begin{align}
\label{eq:ml_sign_consistency_zeta12}
\zeta_{1, 2} \leq
\exp\Bracket{\frac{n}{2}(\log \nfrac{1}{2} - \nfrac{1}{2} + 1)} = e^{-c_1n}.
\end{align}
\par \noindent
\textbf{Analysis of $\zeta_2$.}
With the union bound, we have
\begin{align}
\label{eq:ml_sign_consistency_zeta2}
\zeta_2
\leq~&
\underbrace{\Prob\bracket{
2\langle\Proj_{\bX_S}^{\perp} \bX_{\calT_1}\bbeta^{\natural}_{\calT_1}, \bw \rangle \lsim
-\sigma \sqrt{k\log p}  \|\Proj_{\bX_S}^{\perp} \bX_T\bbeta^{\natural}_T\|_2}}_{
\defequal~\zeta_{2, 1}} \notag\\
+~& \underbrace{
\Prob\bigg(\|\Proj_{\bX_S}^{\perp} \bX_{\calT_1} \bbeta^{\natural}_{\calT_1}\|_2
\geq \sqrt{3(n-k)}
\|\bbeta^{\natural}_{\calT_1}\|_2\bigg)}_{\defequal ~\zeta_{2,2}}.
\end{align}
For $\zeta_{2,1}$, we exploit the independence
between $\bX$ and $\bw$. Conditional on
$\bX$, we can view
$2\langle\Proj_{\bX_S}^{\perp} \bX_{\calT_1}\bbeta^{\natural}_{\calT_1}, \bw \rangle$ as a Gaussian random variable with zero mean
and $4\sigma^2\|\Proj_{\bX_S}^{\perp} \bX_{\calT_1}\bbeta^{\natural}_{\calT_1}\|^2_2$ variance, i.e.,
$\normdist\bracket{0, 4\sigma^2\|\Proj_{\bX_S}^{\perp} \bX_{\calT_1}\bbeta^{\natural}_{\calT_1}\|^2_2}$.
Then, we have
\begin{align}
\label{eq:ml_sign_consistency_zeta21}	
\zeta_{2,1}\leq \exp\bracket{-\frac{c\sigma^2 (k\log p)\|\Proj_{\bX_S}^{\perp} \bX_{\calT_1}\bbeta^{\natural}_{\calT_1}\|^2_2}{4\sigma^2\|\Proj_{\bX_S}^{\perp} \bX_{\calT_1}\bbeta^{\natural}_{\calT_1}\|^2_2}}
= e^{-ck\log p}.
\end{align}
For $\zeta_{2,2}$, we follow a similar proof as in
bounding $\zeta_{1}$ and have
\begin{align}
\label{eq:ml_sign_consistency_zeta22}
\zeta_{2,2}\leq~&
\Prob\big(\|\Proj_{\bX_S}^{\perp} \bX_{\calT_1} \bbeta^{\natural}_{\calT_1}\|_2^2
\geq  3(n-k)\|\bbeta^{\natural}_{\calT_1}\|^2_2,
\|\bX_{\calT_1} \bbeta^{\natural}_{\calT_1} \|_2 \leq \sqrt{\nfrac{3n}{2}}\|\bbeta^{\natural}_{\calT_1}\|_2 \big) \notag \\
+~& \Prob\bracket{\|\bX_{\calT_1} \bbeta^{\natural}_{\calT_1} \|_2 \geq \sqrt{\nfrac{3n}{2}}\|\bbeta^{\natural}_{\calT_1}\|_2} \notag \\
\leq ~&
\Prob\bracket{\|\Proj_{\bX_S}^{\perp} \bX_{\calT_1} \bbeta^{\natural}_{\calT_1}\|_2^2
\geq \frac{2(n-k)}{n}\|\bX_{\calT_1}\bbeta^{\natural}_{\calT_1}\|^2_2} + \Prob\bracket{\|\bX_{\calT_1} \bbeta^{\natural}_{\calT_1} \|_2 \geq \sqrt{\nfrac{3n}{2}}\|\bbeta^{\natural}_{\calT_1}\|_2}
\leq
2e^{-c_1 n}.
\end{align}
\par \noindent
\textbf{Analysis of $\zeta_3$.}
Due to the assumptions in Lemma~\ref{lemma:ml_sign_consistency},
we can verify
\[
\nfrac{(n-k)}{4}
\|\bbeta^{\natural}_{\calT_1}\|^2_2
- c\sigma \sqrt{(k\log p)(n-k)}\|\bbeta^{\natural}_{\calT_1}\|_2 \geq~& \min_{i\in T} |\bbeta^{\natural}_i|
\bracket{\nfrac{(n-k)}{4}\min_{i\in T} |\bbeta^{\natural}_i|
- c\sigma\sqrt{(k\log p)(n-k)}} \\
\geq~& c^{'}\sigma^2\sqrt{k\log(\nfrac{ep}{k})}
\Bracket{\sqrt{k\log(\nfrac{ep}{k}})\vcup \sqrt{n-k}}.
\]
Thus, we have
\[
\zeta_3 \leq
\Prob\bracket{|\fnorm{\Proj_{\bX_T}^{\perp} \bw}^2 -
\fnorm{\Proj_{\bX_S}^{\perp} \bw}^2|\geq \Delta},
\]
where $\Delta$ is set as
$c^{'}\sigma^2\sqrt{k\log(\nfrac{ep}{k})}
\Bracket{\sqrt{k\log(\nfrac{ep}{k}})\vcup \sqrt{n-k}}$.
Noticing the relation
$\Expc \fnorm{\Proj_{\bX_T}^{\perp} \bw}^2 = \Expc
\fnorm{\Proj_{\bX_S}^{\perp} \bw}^2$, we obtain
\[
& \Prob\bracket{|\fnorm{\Proj_{\bX_T}^{\perp} \bw}^2 -
\fnorm{\Proj_{\bX_S}^{\perp} \bw}^2|\geq \Delta}\leq 2\Prob\bracket{|\fnorm{\Proj_{\bX_T}^{\perp} \bw}^2 - \Expc\fnorm{\Proj_{\bX_T}^{\perp} \bw}^2|
\geq \nfrac{\Delta}{2}}.
\]
Due to the independence between $\bX$ and
$\bw$, we can view $\nfrac{\fnorm{\Proj_{\bX_T}^{\perp} \bw}^2}{\sigma^2}$ as a $\chi^2$ random variable with
freedom $n-k$. Thus, we conclude
\begin{align}
\label{eq:ml_sign_consistency_zeta3}
\zeta_3 \leq~& 2\Prob\bracket{|\fnorm{\Proj_{\bX_T}^{\perp} \bw}^2 - \Expc\fnorm{\Proj_{\bX_T}^{\perp} \bw}^2|
\geq \nfrac{\Delta}{2}}
\leq 4\exp\bracket{-\bracket{\frac{\Delta}{16\sigma^2 }\vcap \frac{\Delta^2}{32\sigma^4 (n-k)} }}
\leq 4 e^{-ck\log p}.
\end{align}
Combining \eqref{eq:ml_sign_consistency_tot},
\eqref{eq:ml_sign_consistency_zeta1},
\eqref{eq:ml_sign_consistency_zeta11},
\eqref{eq:ml_sign_consistency_zeta12},
\eqref{eq:ml_sign_consistency_zeta2},
\eqref{eq:ml_sign_consistency_zeta21},
\eqref{eq:ml_sign_consistency_zeta22},
and \eqref{eq:ml_sign_consistency_zeta3}
yields
\begin{align}
\label{eq:ml_sign_consistency_single}
\Prob\bracket{\fnorm{\Proj_{\bX_S}^{\perp} \by}
\leq \fnorm{\Proj_{\bX_T}^{\perp} \by}}
\leq c_0 e^{-c_1k\log p},
\end{align}
where the assumption $n\gsim k\log p$ is invoked.
Notice that \eqref{eq:ml_sign_consistency_single}
is w.r.t. a fixed index set $S$. In the end, we
iterate over all possible index set $S\neq T$ and
complete the proof with the union bound, which
reads as
\[
& \Prob(S\neq T)
\leq \sum_{S\neq T}\Prob\bracket{\fnorm{\Proj_{\bX_S}^{\perp} \by}
\leq \fnorm{\Proj_{\bX_T}^{\perp} \by}} \lsim {p\choose k}\cdot  e^{-c_1k\log p}
\leq (\nfrac{ep}{k})^k  e^{-c_1k\log p}
= e^{-c_2 k\log p}.
\]
\end{proof}

\begin{lemma}
\label{lemma:depend_proj_small_ball}
We have
\[
& \Prob\bracket{\|\Proj_{\bPi\bX_S}^{\perp}\bPitrue\bX_{T}\bbeta_T^{\natu}\|_2^2
< t\|\bbetatrue_T\|_2^2,~~\exists~\bPi, S~~\textup{s.t.}~~\dH(\bPi, \bPitrue) = h} \notag \\
\leq~& 2n^{-2h} + \
6 \exp\bracket{\frac{h}{10}\bracket{\log\bracket{\frac{4e^6 n^{16h/n}t}{h}} -\frac{4e^6 n^{16h/n}t}{h}+ 1}},
\]
where $t < h/(4e^6 n^{16h/n})$, and $h\geq 2$.
\end{lemma}

\begin{proof}
We assume $\bPitrue = \bI$ w.l.o.g. With some
simple algebraic manipulations, we have
\[
\big\|\Proj_{\bX_{S\bigcup \calT_1}}^{\perp} \bPi^{\rmt}\bX_T\bbetatrue_T\big\|_{2}^2= ~&
\big\|\Proj_{{\bPi\bX_{S\bigcup \calT_1}} }^{\perp}\bX_T\bbetatrue_T\big\|_{2}^2 \leq  \big \|\Proj_{\bPi\bX_S}^{\perp}\bX_{T}\bbetatrue_T \big \|_2^2.
\]
Then, we obtain
\begin{align}
\label{eq:depend_proj_small_ball_tot}
& \Prob\big(\|\Proj_{\bPi\bX_S}^{\perp}\bX_{T}\bbeta_T^{\natu}\|_2^2
< t\|\bbetatrue_T\|_2^2,~~\exists~\bPi, S~~\textup{s.t.}~~\dH(\bPi, \bPitrue) = h\big)\notag \\
\leq &~ \hspace{-0.02in}\Prob\big(
\big\|\Proj_{\bX_{S\bigcup \calT_1}}^{\perp} \bPi^{\rmt}\bX_T\bbetatrue_T\big\|_{2}^2
\leq t\|\bbetatrue_T\|_2^2,~\exists~\bPi, S~\textup{s.t.}~\dH(\bPi, \bPitrue) = h\big)
\leq  \zeta_1 + \zeta_2,
\end{align}
where $\zeta_1$ and $\zeta_2$ are defined as
\[
\zeta_1 \defequal ~& \Prob\bigg(\norm{\Proj_{\bX_{T\bigcup \calT_3}}^{\perp} \
\Proj_{\bX_T}^{\perp}\bPi\bX_T\bbetatrue_T}{2}^2 \leq \gamma_h \norm{\Proj_{\bX_T}^{\perp}\bPi\bX_T\bbetatrue_T}{2}^2 ,  \dH(\bI, \bPi)= h, \exists~S\bigg); \\
\zeta_2 \defequal~& \Prob\bigg(\norm{\Proj_{\bX_T}^{\perp}\bPi\bX_T\bbetatrue_T}{2}^2 \leq \frac{t}{\gamma_h}\|\bbetatrue_T\|^2_2,~\dH(\bI, \bPi) = h, \exists~S\bigg).
\]
Here we set $\gamma_h$ as
$2^{-1}e^{-5}n^{-8h/n} < {2e^5}^{-1} < \frac{n-2k}{n-k}$.
The following context separately bound $\zeta_1$ and
$\zeta_2$.

\vspace{0.1in} \noindent \textbf{Analysis of $\zeta_1$.}
When $\calT_3 = \emptyset$, we can verify $\zeta_1 = 0$.
Then we turn to the case where $\calT_3 \neq \emptyset$.
Conditional on $\bX_T$, we can view $\Proj_{\bX_{T\bigcup \calT_3}}^{\perp}$
as a random projection from a linear space with dimension $n-k$
to a linear space with dimension $n-|T\bigcup \calT_3| \geq n - 2k$.
Invoking Lemma~$2.2$ in~\citet{dasgupta2003elementary}
(listed as Lemma~\ref{lemma:random_proj}), we have
\begin{align}
\label{eq:depend_proj_small_ball_zeta1}
\zeta_1 \stackrel{\cirone}{\leq}~&{p\choose k} \exp\Bracket{\frac{n- 2k}{2}\bracket{\log\bracket{\frac{\gamma_h(n-k)}{n-2k}} + 1 }}
\stackrel{\cirtwo}{\leq} \bracket{\frac{ep}{k}}^k  \
\exp\bracket{\frac{n}{4}\log\bracket{\frac{n-k}{n-2k}\frac{n^{-8h/n}}{2e^4} }  }
\stackrel{\cirthree}{\leq} n^{-2h},
\end{align}
where $\cirone$ is because of the union bound;
$\cirtwo$ is due to  ${p\choose k} \leq (\nfrac{ep}{k})^k$
and the definition of
$\gamma_h$; and
$\cirthree$ is because of the assumption
$n \geq k\log(ep/k) \geq 4k$.

 \vspace{0.1in}\noindent
\textbf{Analysis of $\zeta_2$.}
Due to the independence between $S$ and $T$,
we can safely drop $\exists S$ in $\zeta_{2}$.
The following analysis is
a replication of the proof of Lemma~$3$ in~\citet{pananjady2016linear}
with the only
difference in the parameter setting, namely, $t/\gamma_h$.
We present it only for the
sake of self-containing without claiming any novelties.

Without loss of generality, we assume $T$ to be the
first $k$ entries. With the union bound, we obtain
\[
\zeta_2 \leq~& \Prob\bracket{\|\Proj_{
\bX_{1:k}}^{\perp} \bPi \bX_1\|^2_2 \leq \nfrac{t}{\gamma_h},~\dh(\bI, \bPi) = h}\\
\leq~& \underbrace{\Prob\bracket{\
\norm{\Proj_{\bX_{1:k}}^{\perp}
\Proj_{\bX_1}^{\perp}\bPi \bX_1}{2}^2   \leq \vartheta_h \norm{\Proj_{\bX_1}^{\perp}\bPi \bX_1 }{2}^2, ~\dH(\bI, \bPi) = h}}_{\defequal~\zeta_{2, 1}} \\
+~&
\underbrace{\Prob\bracket{\norm{\Proj_{\bX_1}^{\perp}\bPi \bX_1}{2}^2\leq \frac{t}{\gamma_h \vartheta_h},~\dH(\bI, \bPi) = h}}_{\defequal~\zeta_{2, 2}},
\]
where $\gamma_h$ is a positive constant set as $n^{-8h/n}/(2e)$.

For $\zeta_{2, 1}$, we notice the relation $\Proj_{\bX_{1:k}}^{\perp} = \Proj_{\bX_1^{\perp} \bigcap \bX_{2:k}^{\perp}}$.
Condition on $\bX_1$, we can view
$\Proj_{\bX_{1:k}}^{\perp}$ as a random
projection from a $(n-1)$-dimensional linear space to a
$(n-k)$-dimensional linear space, which yields
\begin{align}
\zeta_{2, 1} \leq &~\hspace{-0.01in}\exp\Bracket{\frac{n-k}{2}\bracket{\log\bracket{\frac{(n-1)\vartheta_h}{n-k}} - \frac{(n-1)\vartheta_h}{n-k} + 1}}
\leq \hspace{-0.01in} n^{-2h},
\label{eq:ml_small_ball_lemma_zeta2_part1}
\end{align}
where $\vartheta_h \leq \frac{n-k}{n-1}$ is due to
Lemma $2.2$ in~\citet{dasgupta2003elementary} (also listed as
Lemma~\ref{lemma:random_proj}).

For $\zeta_{2,2}$, we first perform decomposition
\[
& \norm{\Proj_{\bX_1}^{\perp}\bPi\bX_1}{2}^2
= \norm{\bX_1}{2}^2 - \frac{\la \bX_1, \bPi\bX_1\ra^2}{\norm{\bX_1}{2}^2}
\stackrel{\cirfour}{\geq} \norm{\bX_1}{2}^2 - \abs{\la \bX_1, \bPi \bX_1\ra}
= \frac{1}{2}\Bracket{\norm{\bX_1 - \bPi\bX_1}{2}^2 \vcap \norm{\bX_1 + \bPi\bX_1}{2}^2},
\]
where in $\cirfour$ we use the Cauchy inequality
such that $\abs{\la \bX_1, \bPi\bX_1\ra} \leq \|\bX_1\|_{2}\|\bPi \bX_1\|_{2} = \
\|\bX_1\|_{2}^2$.
When $\|\bX_1 - \bPi\bX_1\|_{2} \vcap \|\bX_1 + \bPi\bX_1\|_{2} = \|\bX_1 - \bPi\bX_1\|_{2}$,
we can directly invoke Lemma~\ref{lemma:permute_cancel_tail_bound} to
bound $\zeta_{2,2}$. Following a similar procedure,
we can show
\begin{align}
\label{eq:ml_small_ball_lemma_zeta2_part2}
\zeta_{2,2}\leq
6\exp\bracket{\frac{h}{10}\bracket{\log\bracket{\frac{t}{h\gamma_h \vartheta_h}} -\frac{t}{h\gamma_h \vartheta_h} + 1}},
\end{align}
provided that we have $t < h\gamma_h \vartheta_h$.
The proof for $\zeta$ is hence completed
by combining
\eqref{eq:depend_proj_small_ball_tot},
\eqref{eq:depend_proj_small_ball_zeta1},
\eqref{eq:ml_small_ball_lemma_zeta2_part1},
and \eqref{eq:ml_small_ball_lemma_zeta2_part2}.

\end{proof}

\begin{lemma}
\label{lemma:permute_cancel_tail_bound}
Denote $h$ as the Hamming distance between
$\bI$ and $\bPi$, i.e., $h \defequal \dH(\bI, \bPi)$.
Assume $\bx \in \RR^n$ be a random
vector satisfying $\bx \sim \normdist\bracket{\bZero, \bI_{n\times n}}$, then we have
\[
\Prob\bracket{\norm{\bracket{\bI- \bPi} \bx}{2}^2 \leq \vartheta} \leq \
6\exp\Bracket{\frac{h}{10}\bracket{\log\bracket{\frac{\vartheta}{2h}} -\frac{\vartheta}{2h} + 1}},
\]
for $\vartheta \leq 2h$.
\end{lemma}

\begin{proof}
Adopting the similar proof tricks as in~\citet{pananjady2016linear}, we separately consider
the two cases where $h = 2$ and $h\geq 3$.

\vspace{0.1in}\noindent
\textbf{(Case I) $h = 2$.} We assume the first two rows are switched w.l.o.g.
Then we have
\[
\Prob\bracket{\norm{\bracket{\bI- \bPi} \bx}{2}^2 \leq \vartheta} =~& \
\Prob\big[\bracket{x_{1} - x_{2}}^2 \leq \nfrac{\vartheta}{2}\big]
\stackrel{\cirone}{\leq} \exp\bracket{-\nfrac{1}{2}\bracket{\nfrac{\vartheta}{4} - \log\bracket{\nfrac{\vartheta}{4}} - 1}},
\]
where in $\cirone$ we use the tail bounds for the $\chi^2$ random variable
$(x_{1} - x_{2})^2/2$ with freedom $1$.

\vspace{0.1in}\noindent
\textbf{(Case II) $h\geq 3$.}
We decompose the non-zero rows
of $\bracket{\bPi - \bPitrue}$ into three disjoint sets
$\calI_{\ell}$ $(1\leq \ell \leq 3)$ such that
$(i)$ the cardinality of each set $\calI_{\ell}$ is lower bounded by
$\lfloor h/3 \rfloor$, i.e., $\abs{\calI_{\ell}} = h_{\ell} \geq \lfloor h/3 \rfloor$;
and $(ii)$ we have $j$ and $\pi(j)$ reside within
different sets for an arbitrary index $j$, where $\pi(\cdot)$ denotes the permutation map pertaining
to $\bPi$.

Define $Z_{\ell} = \sum_{j\in \calI_{\ell}}\nfrac{(x_j - x_{\pi(j)})^2}{2}$,
which is a $\chi^2$ random variable with freedom $h_{\ell}$.
Then, we can decompose
$\norm{(\bI - \bPi)\bx}{2}^2 = 2(Z_1 + Z_2 + Z_3)$
and obtain
\[
\Prob\bracket{\norm{\bracket{\bI - \bPi} \bx}{2}^2 \leq \vartheta} \leq~&
\sum_{\ell =1}^3 \Prob\bracket{Z_{\ell} \leq \nfrac{h_{\ell}\vartheta}{(2h)}}
\leq \sum_{\ell=1}^3 \
\exp\bracket{-\frac{h_{\ell}}{2}\bracket{\frac{\vartheta}{2h} - \log\bracket{\frac{\vartheta}{2h}} - 1 }} \\
\stackrel{\cirtwo}{\leq}~& 6\exp\bracket{-\frac{h}{10}\bracket{\frac{\vartheta}{2h} - \log\bracket{\frac{\vartheta}{2h}} - 1 }},
\]
where in $\cirtwo$ we use the relation $h_i \geq \lfloor h/3 \rfloor$.
The proof is completed by summarizing the above two cases.
\end{proof}

\newpage

\section{Proof of Theorem~\ref{thm:robust_lasso_permutation}}

First, we restate the definition of
$(\wh{\bbeta}, \wh{\bXi})$, which is written as
\[
(\wh{\bXi}, \wh{\bbeta}) = \argmin_{\bXi, \bbeta}
~&\frac{1}{2n}\norm{\by - \bX\bbeta - \sqrt{n}\cdot \bXi}{2}^2 + \lambda_{\bXi} \norm{\bXi}{1}
+ \lambda_{\bbeta} \norm{\bbeta}{1}.
\]
Then, we would like to prove  Theorem~\ref{thm:robust_lasso_permutation}.

\begin{proof}
The proof is a combination of
\citep{nguyen2013robust} and
\citep{slawski2017linear}.
Define $\bu \defequal \wh{\bbeta} - \bbeta^{\natural}$ and
$\bv \defequal \wh{\bXi} - \nfrac{(\bI- \bPitrue)\bX\bbeta^{\natural}}{\sqrt{n}}$.
In addition, we define the support set of $\bbetatrue$ and
$(\bI -\bPitrue)\bX\bbetatrue$ as $T$ and $S$, respectively.
According to our definition, their
cardinality is bounded by $k$ and $h$, respectively, namely,
$|T|\leq k$ and $|S| \leq h$.
Before delving into the technical details,
we first illustrate the proof outline.
\begin{itemize}
\item
\textbf{Step I.}	
According to the optimality of
\eqref{eq:robust_lasso_optim_def}, we show
\begin{align}
\label{eq:ht_upper_bound}
\|\bu \|_1 \leq
4\sqrt{k}\norm{\bu}{2}
+  \nfrac{3\sqrt{h}\xilambda}{\belambda}
\norm{\bv}{2};  \\
\|\bv \|_1 \leq
\nfrac{3\sqrt{k}\belambda}{\xilambda}
\norm{\bu}{2}
+ 4\sqrt{h}\norm{\bv}{2}.
\label{eq:gs_upper_bound}
\end{align}
\item
\textbf{Step II.}
We establish the inequality
\begin{align}
\label{eq:l2_quad_upper_bound}
& \bracket{4\sqrt{k}\|\bu\|_2
+ \nfrac{3\sqrt{h}\xilambda}{\belambda}\|\bv\|_2}^2 \lsim \bracket{k\belambda \vcup h\xilambda}\cdot
\Bracket{4\sqrt{k}\|\bu\|_2 +\nfrac{3\sqrt{h}\xilambda}{\belambda} \|\bv\|_2},
\end{align}
and then obtain the upper-bound
$c\cdot k\belambda$ for the
reconstruction error $4\sqrt{k}\|\bu\|_2
+ \nfrac{3\sqrt{h}\xilambda}{\belambda}\|\bv\|_2$.

\item
\textbf{Step III.}
We upper-bound $\norm{\bX\bu}{\infty}$
as $k\sigma \sqrt{\nfrac{(\log p)(\log np)}{n} }$ and complete the proof by invoking
Lemma~\ref{lemma:permute_recover_slawski}.
\end{itemize}
The technical details are presented as follows.
According to the definition of \eqref{eq:robust_lasso_optim_def}, we have
\[
& \frac{1}{2n}\Fnorm{\by - \bX \wh{\bbeta} - \sqrt{n}\cdot \wh{\bXi}}^2
+ \belambda \|\wh{\bbeta}\|_{1}
+ \xilambda \|\wh{\bXi} \|_{1} \leq
\frac{1}{2n}\Fnorm{\by - \bX \betatrue - \sqrt{n}\bXi^{\natural}}^2
+ \belambda \norm{\bbetatrue}{1}
+ \xilambda \norm{\bXi^{\natural}}{1}.
\]
With some standard algebraic manipulation, we
obtain
\begin{align}
\label{eq:robust_lasso_optim_inequal}
\frac{1}{2n}\Fnorm{\bX\bu+ \sqrt{n}\bv }^2
\leq~& \frac{\la \bw, \bX \bu + \sqrt{n}\bv \ra}{n}
+ \underbrace{\belambda\big(\norm{\bbetatrue}{1} - \|\wh{\bbeta}\|_{1}\big)}_{\defequal~\theta_1}
+ \underbrace{\xilambda \big(\norm{\bXi^{\natural}}{1} - \|\wh{\bXi}\|_{1}\big)}_{\defequal~\theta_2} \notag \\
\leq~& \frac{\norm{\bX^{\rmt}\bw}{\infty}}{n}
\cdot \|\bu \|_1 + \frac{\norm{\bw}{\infty}}{\sqrt{n}}
\cdot \|\bv \|_1 + \theta_1 + \theta_2.
\end{align}
For $\theta_1$, we exploit the fact
$\|\bbeta^{\natural}_T \|_1 = \|\bbeta^{\natural}\|_1$ and have
\[
\theta_1 =~& \belambda(\|\bbeta^{\natural}_T \|_1  - \|\wh{\bbeta}_T\|_1 - \|\wh{\bbeta}_{T^c}\|_1) \leq \belambda(\|\bbeta^{\natural}_T - \wh{\bbeta}_T\|_1 - \|\wh{\bbeta}_{T^c}\|_1)
= \belambda(\|\bu_T\|_1 - \|\bu_{T^c}\|_1).
\]
Similarly, we have
$\theta_2 \leq \xilambda(\|\bv_S \|_1 - \|\bv_{S^c}\|_1) $.
According to Lemma~\ref{lemma:xw} and Lemma~\ref{lemma:winf},
we have $\belambda \geq \frac{2\|\bX^{\rmt}\bw\|_{\infty}}{n}$ and $\xilambda \geq \frac{2\|\bw \|_{\infty}}{\sqrt{n}}$.
Summing the above together, we have
\begin{align}
\label{eq:robust_lasso_optim_inequal_nonnegative}
\frac{1}{2n}\Fnorm{\bX\bu+ \sqrt{n}\bv }^2
\leq~&
 \frac{3\belambda}{2}\norm{\bu_T}{1}
+ \frac{3\xilambda}{2}\norm{\bv_S}{1} - \frac{\belambda}{2}\|\bu_{T^c}\|_1
- \frac{\xilambda}{2} \|\bv_{S^c}\|_1.
\end{align}

\noindent \textbf{Step I.}
Notice that
the left-hand size of
\eqref{eq:robust_lasso_optim_inequal_nonnegative}
is non-negative, we obtain
\[
\norm{\bu}{1} =~&
\|\bu_T\|_1 +
\|\bu_{T^c}\|_1\leq
4\norm{\bu_T}{1}
+ \frac{3\xilambda}{\belambda}\norm{\bv_S}{1}
- \underbrace{\frac{\xilambda}{\belambda}\|\bv_{S^c}\|_1}_{\geq 0} \notag \\
& \stackrel{\cirone}{\leq}
4\sqrt{k}\|\bu_T\|_2 +
\nfrac{3\sqrt{h}\xilambda}{\belambda} \|\bv_S\|_2 \leq 4\sqrt{k}\|\bu\|_2
+ \nfrac{3\sqrt{h}\xilambda}{\belambda} \|\bv\|_2,
\]
where in $\cirone$ we exploit the fact such that
$\bu_T$ and $\bv_S$ are $k$-sparse and $h$-sparse
respectively.
As for \eqref{eq:gs_upper_bound}, we follow
a similar approach and finish its proof.

\vspace{0.1in}\noindent \textbf{Step II.}
Without loss of generality, we assume $\belambda\sqrt{k} \geq \xilambda \sqrt{h}$ and have
\begin{align}
\label{eq:robust_lasso_optim_inequal_nonnegative_uv_upper_bound}
& \bracket{4\sqrt{k}\norm{\bu}{2} +
\nfrac{3\sqrt{h}\xilambda}{\belambda} \norm{\bv}{2}}^2  \asymp \belambda^{-2}\bracket{
\belambda\sqrt{k}\norm{\bu}{2} +
\sqrt{h}\xilambda\norm{\bv}{2}}^2 \leq k
\bracket{\|\bu\|_2 + \|\bv\|_2}^2.
\end{align}
On one hand, we can invoke Lemma~\ref{lemma:robust_lasso_xhg_lower_bound} and
obtain an upper-bound for $\bracket{\|\bu\|_2 + \|\bv\|_2}^2$ reading
as
\begin{align}
\label{eq:robust_lasso_optim_inequal_nonnegative_xuv_upper_bound}
\frac{1}{2n}\Fnorm{\bX\bu+ \sqrt{n}\bv }^2
\gsim \bracket{\|\bu \|_2 + \|\bv \|_2}^2.
\end{align}
On the other hand,
 \eqref{eq:robust_lasso_optim_inequal_nonnegative} yields
the upper-bound for $\frac{1}{2n}\Fnorm{\bX\bu+ \sqrt{n}\bv }^2$,
which can be written as
\begin{align}
\label{eq:robust_lasso_optim_inequal_nonnegative_xuv_upper_bound2}
\frac{1}{2n}\Fnorm{\bX\bu+ \sqrt{n}\bv }^2
\leq~&
 \frac{3\belambda}{2}\norm{\bu_T}{1}
+ \frac{3\xilambda}{2}\norm{\bv_S}{1} \leq \frac{3\belambda}{2}\sqrt{k}\norm{\bu}{2}
+ \frac{3\xilambda}{2}\sqrt{h}\norm{\bv}{2}.
\end{align}
Combining \eqref{eq:robust_lasso_optim_inequal_nonnegative_uv_upper_bound},
\eqref{eq:robust_lasso_optim_inequal_nonnegative_xuv_upper_bound}, and \eqref{eq:robust_lasso_optim_inequal_nonnegative_xuv_upper_bound2}
then yields the relation
\[
& \bracket{4\sqrt{k}\norm{\bu}{2} +
\nfrac{3\sqrt{h}\xilambda}{\belambda} \norm{\bv}{2}}^2
\lsim \frac{k}{n}\Fnorm{\bX\bu + \sqrt{n}\bv}^2 \lsim k\bracket{4\belambda\sqrt{k}\|\bu \|_2
+ 3\xilambda \sqrt{h} \|\bv \|_2}.
\]
Dividing both sides by
$4\sqrt{k}\norm{\bu}{2} +
\nfrac{3\sqrt{h}\xilambda}{\belambda} \norm{\bv}{2}$ then completes the proof
of \eqref{eq:l2_quad_upper_bound}.

\vspace{0.1in}\noindent
\textbf{Step III.}
Our goal is upper-bound
$\norm{\bX\bu}{\infty}$, which reads as
\[
\norm{\bX\bu}{\infty}
=~& \max_{i} \abs{\la \bX_{i, :}, \bu\ra }
\leq \max_{ij} |\bX_{ij}|\cdot \norm{\bu}{1}
\stackrel{\cirtwo}{\lsim} \sqrt{\log np}
\bracket{\sqrt{k} \|\bu\|_2 + \nfrac{\sqrt{h}\xilambda}{\belambda}\|\bv \|_2 } \\
\stackrel{\cirthree}{\lsim}~& \sigma \sqrt{\log np} \bracket{k\sqrt{\frac{\log p}{n}}
\vcup h \sqrt{\frac{\log n}{n}}
},
\]
where in $\cirtwo$ we condition on the event
$\max_{ij}|\bX_{ij}|\lsim \sqrt{\log np}$,
and in $\cirthree$ we use the inequality in
\eqref{eq:l2_quad_upper_bound}.
Provided that
\[
\snr \gsim \frac{n^{2(1+\varepsilon)} (n-1)^2}{4\pi}
\Bigg[& \sqrt{\log np} \bracket{k\sqrt{\frac{\log p}{n}}
\vcup h \sqrt{\frac{\log n}{n}}
} +2\log(n^{1+\varepsilon}(n-1)) \Bigg]^2,
\]
we invoke Lemma~\ref{lemma:permute_recover_slawski}
and complete the proof such that ground-truth permutation $\bPitrue$
can be obtained  with probability exceeding
$1-2n^{-\varepsilon}$.
\end{proof}

\subsection{Supporting Lemmas}
This subsection collects the supporting lemmas and useful facts used in the proof thereof.
\begin{lemma}
\label{lemma:max_xij}
We have $\max_{1\leq i \leq n, 1\leq j\leq p}|\bX_{ij}| \leq 2\sqrt{\log np}$ with probability exceeding
$1 - (np)^{-1}$.	
\end{lemma}

\begin{lemma}
\label{lemma:xw}
We have $\norm{\bX^{\rmt}\bw}{\infty}\lsim \sigma \sqrt{n\log p}$ with
probability exceeding $1-c_0 p^{-c_1}$.
\end{lemma}

\begin{lemma}
\label{lemma:winf}	
We have $\norm{\bw}{\infty}\lsim \sigma\sqrt{\log n}$ with
probability $1-c_0 n^{-c_1}$.
\end{lemma}

 Since the above results are quite standard, we list
them without giving a detailed proof.

\begin{lemma}[Lemma~$1$ In~\citet{nguyen2013robust}]
\label{lemma:robust_lasso_xhg_lower_bound}
Consider the optimal solution
$(\wh{\bbeta}, \wh{\bXi})$ to \eqref{eq:robust_lasso_optim_def} with
regularizer coefficients being set as
$\belambda \asymp \sigma \sqrt{\nfrac{\log p}{n}}$
and
$\xilambda \asymp \sigma \sqrt{\nfrac{\log n}{n}}$,
respectively.
Assuming that $n \gsim k\log p$ and
$h \lsim \frac{n}{\log n}$, we have
\[
\frac{1}{\sqrt{n}}\Fnorm{\bX\bu + \sqrt{n}\bv }
\gsim \|\bu \|_{2} + \|\bv \|_2
\]
hold with probability at least $1 - c_0 e^{-c_1n}$, where $c_0, c_1 > 0$ are some fixed positive constants.
\end{lemma}

\begin{lemma}[Theorem~$3$ (Part (a)) in~\citet{slawski2017linear}]
\label{lemma:permute_recover_slawski}
Conditional on the event $\calE_{\wt{\bbeta}}$ such that
$\calE_{\wt{\bbeta}}\defequal \set{\|\bX(\wt{\bbeta}-  \bbeta^{\natural})\|_{\infty} \leq \sigma \Delta}$,
we reconstruct the permutation matrix via
$\wh{\bPi} = \argmax_{\bPi} \langle \bY, \bPi \bX \wt{\bbeta}\rangle$.
Provided that
\[
\snr > \frac{n^2(n-1)^2}{4\delta^2 \pi}
\Bracket{\Delta + 2\log\frac{n(n-1)}{\delta}}^2,
\]
Then, we can obtain the ground-truth permutation
with probability exceeding $1-2\delta$,
i.e., $\Prob(\wh{\bPi} = \bPitrue|\calE_{\wt{\bbeta}}) \geq 1-2\delta$.
\end{lemma}

\section{Useful Facts About Probability Inequalities}

For the self-containing of this paper, we list some useful facts
about probability inequalities in this section.

\begin{lemma}[Lemma 2.2 In~\citet{dasgupta2003elementary}]
\label{lemma:random_proj}
For a projection matrix $\Proj_{d_1 \rightarrow d_2}$ which
projects a fixed vector $\bZ\in \RR^{d_1}$ to a uniformly random
subspace with dimension $d_2$, we have
\[
\Prob\bracket{\norm{\Proj_{d_1\rightarrow d_2}\bZ}{2}^2\leq \frac{\alpha d_2}{d_1}\norm{\bZ}{2}^2 } \hspace{-0.02in}
&\leq \exp\big(\frac{d_2}{2}\bracket{\log \alpha - \alpha + 1} \big),\alpha < 1; \\
\Prob\bracket{\norm{\Proj_{d_1\rightarrow d_2}\bZ}{2}^2\geq \frac{\alpha d_2}{d_1}\norm{\bZ}{2}^2 } \hspace{-0.02in}
&\leq \exp\big(\frac{d_2}{2}\bracket{\log \alpha - \alpha + 1} \big),\alpha > 1.
\]
\end{lemma}

\end{appendices}

\end{document}